\documentclass[10pt,draftclsnofoot,onecolumn]{IEEEtran}  

\usepackage{amsmath,graphics,graphicx}
\usepackage{amssymb}
\usepackage{dsfont}

\IEEEoverridecommandlockouts

\newcommand{\bea}{\begin{eqnarray}}
\newcommand{\eea}{\end{eqnarray}}

\begin{document}

\title{On the Entropy Region of Gaussian Random Variables*\footnote{*Preliminary versions of this manuscript have appeared in \cite{hash2008} and \cite{allertonShHa}.   }}


%

\author{\authorblockN{Sormeh Shadbakht and Babak Hassibi\thanks{
This work was supported in part by the National Science Foundation under grants CCF-0729203, CNS-0932428 and CCF-1018927, by the Office of Naval Research under the MURI grant N00014-08-1-0747, and by Caltech's Lee Center for Advanced Networking.
}}\\
\authorblockA{Department of Electrical Engineering\\
California Institute of Technology\\
Pasadena, CA 91125\\
sormeh, hassibi@caltech.edu}}

\maketitle

\def\real{{\mathchoice%
{\hbox{\rm\setbox1=\hbox{I}\copy1\kern-.45\wd1 R}}
{\hbox{\rm\setbox1=\hbox{I}\copy1\kern-.45\wd1 R}}
{\hbox{\scriptsize\rm\setbox1=\hbox{I}\copy1\kern-.45\wd1 R}}
{\hbox{\scriptsize\rm\setbox1=\hbox{I}\copy1\kern-.45\wd1 R}}}}
\def\Zint{{\mathchoice{\setbox1=\hbox{\sf Z}\copy1\kern-.75\wd1\box1}
{\setbox1=\hbox{\sf Z}\copy1\kern-.75\wd1\box1}
{\setbox1=\hbox{\scriptsize\sf Z}\copy1\kern-.75\wd1\box1}
{\setbox1=\hbox{\scriptsize\sf Z}\copy1\kern-.75\wd1\box1}}}
\newcommand{\complex}{ \hbox{\rm C\kern-0.45em\rule[.07em]{.02em}{.58em}%
\kern 0.43em}}

\def\IC{{\bf C}}
\def\IN{{\bf N}}
\def\IZ{{\bf Z}}
\def\ID{{\bf D}}
\def\IT{{\bf T}}

\def\A{{\cal A}}
\def\B{{\cal B}}
\def\C{{\cal C}}
\def\D{{\cal D}}
\def\E{{\cal E}}
\def\F{{\cal F}}
\def\G{{\cal G}}
\def\I{{\cal I}}
\def\J{{\cal J}}
\def\K{{\cal K}}
\def\L{{\cal L}}
\def\M{{\cal M}}
\def\N{{\cal N}}
\def\O{{\cal O}}
\def\P{{\cal P}}
\def\Q{{\cal Q}}
\def\R{{\cal R}}
\def\S{{\cal S}}
\def\T{{\cal T}}
\def\U{{\cal U}}
\def\V{{\cal V}}
\def\W{{\cal W}}
\def\X{{\cal X}}
\def\Y{{\cal Y}}
\def\Z{{\cal Z}}

\newcommand{\ovr}{\overline}
\newcommand{\no}{\noindent}
\def\ker{{\rm ker}\,}
\def\dim{{\rm dim}\,}
\newcommand{\st}{\stackrel{J}{\approx}}
\newcommand{\stt}{\stackrel{I}{\approx}}
\newcommand{\be}{\begin{equation}}                         
\newcommand{\ee}{\end{equation}}
\newcommand{\beqr}{\begin{eqnarray}}
\newcommand{\eeqr}{\end{eqnarray}}
\newcommand{\beqrx}{\begin{eqnarray*}}
\newcommand{\eeqrx}{\end{eqnarray*}}
\newcommand{\ba}{\left[ \begin{array}}
\newcommand{\ea}{\\ \end{array} \right]}
\newcommand{\bi}{\begin{itemize}}
\newcommand{\ei}{\end{itemize}}
\newcommand{\lb}{\label}
\newcommand{\rf}{\ref}
\newcommand{\ov}{\bar}
\newcommand{\td}{(t-1)}
\newcommand{\xx}{t-1}
\newcommand{\tl}{\tilde}
\newcommand{\bb}{\cite}
\newtheorem{lemma}{Lemma}
\newtheorem{theorem}{Theorem}
\newtheorem{corollary}{Corollary}
\newtheorem{problem}{Problem}
\newtheorem{proposition}{Proposition}
\newtheorem{definition}{Definition}
\newtheorem{fact}{Fact}
\newtheorem{solution}{Solution}
\newtheorem{algorithm}{Algorithm}
\newtheorem{example}{Example}
\newtheorem{result}{Result}
\newtheorem{conjecture}{Conjecture}
\newtheorem{notation}{Notation}
\newcommand{\ad}{{\rm \;\;\;\; and \;\;\;\;\;}}                         
\newcommand{\qd}{\hfill{\qed}}
\def\qed{{\ \vrule width 2.5mm height 2.5mm \smallskip}}

\begin{abstract}
Given $n$ (discrete or continuous) random variables $X_i$, the
$(2^n-1)$-dimensional vector obtained by evaluating the joint entropy of
all non-empty subsets of $\{X_1,\ldots,X_n\}$ is called an entropic
vector. Determining the region of entropic vectors is an important
open problem with many applications in information theory. Recently, it has been shown that the
entropy regions for discrete and continuous random variables, though
different, can be determined from one another. An important class of
continuous random variables are those that are vector-valued and
jointly Gaussian. It is known that Gaussian random variables violate
the Ingleton bound, which many random variables such as those obtained
from linear codes over finite fields do satisfy, and they also achieve
certain non-Shannon type inequalities. 
In this paper we give a full characterization of the {\em convex cone} of the entropy region of three jointly Gaussian
vector-valued random variables and prove that it is the same as the convex cone of three {\it scalar-valued} Gaussian random variables and further that it yields the entire entropy region of 3 {\it arbitrary} random variables. 
We further determine the actual entropy region of 3 vector-valued jointly Gaussian random variables through a conjecture. 
For $n\geq 4$ number of random variables, we point out a set of $2^n-1-\frac{n(n+1)}{2}$ minimal necessary and sufficient conditions that $2^n-1$ numbers must satisfy in order to correspond to the entropy vector of $n$ {\it scalar} jointly Gaussian random variables.
This improves on a result of Holtz and Sturmfels which gave a nonminimal set of conditions.
These constraints are related to Cayley's hyperdeterminant and hence with an eye towards characterizing the entropy region of jointly Gaussian random variables, we also present some new results in this area. We obtain a new (determinant) formula for the $2\times 2\times 2$ hyperdeterminant and we also give a new (transparent) proof of the fact
that the principal minors of an $n\times n$ symmetric matrix satisfy the ${2\times 2\times \ldots \times 2}$ (up to $n$ times) hyperdeterminant relations.
\end{abstract}

\section{Introduction} \label{Gaussian_sec:intro} 
Obtaining the capacity region of information networks has long been an
important open problem. It turns out that there is a fundamental
connection between the entropy region of a number of random variables
and the capacity region of networks \cite{hashITW07} \cite{Yan-capacityregion}. However
determining the entropy region has proved to be an extremely difficult
problem and there have been different approaches towards
characterizing it. While most of the effort has been towards obtaining
outer bounds for the entropy region by determining valid information
inequalities \cite{framework, ma07,non-shannon, six-new, group,
  quasi-uniform, makarychev, zhang-new}
some have focused on innerbounds \cite{hash2007, characterization,
  conditional-independence} which may prove to be more useful since they
yield achievable regions.

Let $X_1, \cdots ,X_n$ be $n$ jointly distributed discrete random
variables with arbitrary alphabet size $N$. The vector of all the $2^n-1$ joint
entropies of these random variables is referred to as their ``entropy
vector" and conversely any $2^n-1$ dimensional vector whose elements
can be regarded as the joint entropies of some $n$ random variables,
for some alphabet size $N$, is called ``entropic''. The \textit{entropy
  region} is defined as the region of all possible entropic
vectors and is denoted by $\Gamma^*_n$ \cite{framework}. Let
$\mathcal{N}= \{1,\cdots,n\}$ and $s, s' \subseteq
\mathcal{N}$. If we define $X_{s} = \{X_i : i \in s\}$ then
it is well known that the joint entropies $H(X_{s})$ (or $H_s$,
for simplicity) satisfy the following inequalities:

\begin{enumerate}
\item $H_\emptyset = 0$
\item For $s\subseteq s'$: $H_s\leq H_{s'}$
\item For any $s,s'$:
  $H_{s\cup s'}+H_{s\cap s'}\leq H_s+H_{s'}$.
\end{enumerate}
These are called the basic inequalities of Shannon information
measures and the last one is referred to as the ``submodularity
property''. They all follow from the nonnegativity of the conditional
mutual information \cite{non-shannon,polymat,uniqueness}. Any
inequality obtained from positive linear combinations of conditional
mutual information is called a ``Shannon-type'' inequality. The space of
all ${2^n-1}$ dimensional vectors which only satisfy the Shannon
inequalities is denoted by $\Gamma_n$. It has been shown that
$\Gamma^*_2 = \Gamma_2$ and $\bar{\Gamma}^*_3 = \Gamma_3$ where
$\bar{\Gamma}^*_3$ denotes the closure of $\Gamma^*_3$
\cite{non-shannon}. However, for $n\geq 4$, in 1998 the first
non-Shannon type information inequality was discovered
\cite{non-shannon} which demonstrated that $\Gamma^*_4$ is strictly
smaller than $\Gamma_4$. Since then many other non-Shannon type
inequalities have been discovered \cite{matus-construction, six-new,
  makarychev,zhang-new}. Nonetheless, the complete characterization of
$\Gamma^*_n$ for $n\geq 4$ remains open.

The effort to characterize the entropy region has focused
on discrete random variables, ostensibly because the study of discrete
random variables is simpler. However, continuous random
variables are as important, where now for any collection of random
variables $X_s$, with joint probability density function
$f_{X_s}(x_s)$, the differential entropy
is defined as
\begin{equation}
h_s = - \int {f_{X_s}(x_s) \log
  f_{X_s}(x_s) dx_s}.
\end{equation}

Let $\sum_{s} \gamma_{s} H_s \geq 0$ be a valid
discrete information inequality. This inequality is called
{\em balanced} if for all $i\in \mathcal{N }$ we have $\sum_{s: i\in
  s} \gamma_s = 0$. Using this notion Chan \cite{chan-balanced}
has shown a correspondence between discrete and continuous information
inequalities, which allows us to compute the entropy region for one
from the other.
\begin{theorem}[Discrete/continuous information inequalities]
\-
\begin{enumerate}
\item A linear continuous information inequality \mbox{$\sum_{s} \gamma_{s}
h_s \geq 0$} is valid if and only if its discrete counterpart
$\sum_{s} \gamma_{s} H_s \geq 0$ is balanced and
valid.
\item
A linear discrete information inequality \mbox{$\sum_{s} \gamma_{s} H_s \geq 0$} is valid if and only if it can be written as $\sum_{s}{\beta_{s} h_{s}} + \sum_{i=1}^{n}r_i(h_{i,i^c}-h_{i^c})$ for some $r_i\geq0$, where $\sum_{s}{\beta_{s} h_{s}}\geq0$ is a valid continuous information inequality ($i^c$ denotes the complement of $i$ in $\mathcal{N}$).
\end{enumerate}
\label{thm:contin-ineq}
\end{theorem}

The above Theorem suggests that one can also study continuous random
variables to determine $\Gamma^*_n$. Among all continuous random
variables, the most natural ones to study first (for many of the
reasons further described below) are Gaussians. This will be the main
focus of this paper.

Let $X_1,\cdots,X_n\in{\mathds{R}}^T$ be $n$ jointly distributed
zero-mean\footnote{Since differential entropy is invariant to shifts
  there is no point in assuming nonzero means for the $X_i$.} vector-valued real Gaussian random variables of vector size $T$ with covariance matrix $R\in{\mathds{R}}^{nT\times
  nT}$. Clearly, $R$ is symmetric, positive semidefinite, and consists
of block matrices of size $T\times T$ (corresponding to each random
variable). We will allow $T$ to be arbitrary and will therefore consider
the {\em normalized} joint entropy of any subset $s\subseteq{\cal
  N}$ of these random variables
\begin{equation}
{\underline h}_s = \frac{1}{T}\cdot \frac{1}{2} \log \Big((2\pi
e) ^{T|s|} \det R_s\Big),
\end{equation}
where $|s|$ denotes the cardinality of the set $s$ and
$R_s$ is the $T|s|\times T|s|$ matrix obtained by
keeping those block rows and block columns of $R$ that are indexed by
$s$. Note that our normalization is by the dimensionality of the
$X_i$, i.e., by $T$, and that we have used $\underline{h}$ to denote
normalized entropy.

Normalization has the following important consequence.

\begin{theorem}[Convexity of the region for $\underline{h}$]
The closure of the region of normalized Gaussian entropy vectors is convex.
\end{theorem}

\begin{proof}
Let $\underline{h}^x$ and $\underline{h}^y$ be
two normalized Gaussian entropy vectors. This means that the first
corresponds to some collection of Gaussian random variables
$X_1,\ldots,X_n\in{\mathds{R}}^{T_x}$ with the covariance matrix $R^x$,
for some $T_x$, and the second to some other collection
$Y_1,\ldots,Y_n\in{\mathds{R}}^{T_y}$ with the covariance matrix $R^y$,
for some $T_y$. Now generate $N_x$ copies of jointly Gaussian random variables $X_1,\ldots,X_n$ and $N_y$ copies of $Y_1,\ldots,Y_n$ and define the new set of random variables $Z_i =
\ba{cccccc} (X_i^1)^t & \ldots & (X_i^{N_x})^t & (Y_i^1)^t & \ldots &
(Y_i^{N_y})^t \ea^t$, where $(\cdot)^t$ denotes the transpose, by stacking $N_x$ and $N_y$ independent copies of
each, respectively, into a $N_xT_x+N_yT_y$ dimensional vector. Clearly
the $Z_i$ are
jointly-Gaussian. Due to the independence of the $X_i^k$ and
$Y_i^l$, $k=1,\ldots N_x$, $l=1,\ldots,N_y$, the non-normalized
entropy of the collection of random variables $Z_s$ is
\[ h^z_s = N_xT_x\underline{h}_s^x+
N_yT_y\underline{h}_s^y. \]
To obtain the normalized entropy we should divide by $N_xT_x+N_yT_y$
\[ \underline{h}^z_s =
\frac{N_xT_x}{N_xT_x+N_yT_y}\underline{h}_s^x +
\frac{N_yT_y}{N_xT_x+N_yT_y}\underline{h}_s^y, \]
which, since $N_x$ and $N_y$ are arbitrary, implies that every vector
that is a convex combination of $\underline{h}^x$ and
$\underline{h}^y$ is entropic and generated by a Gaussian.
\end{proof}

Note that $\underline{h}_s$ can also be written as follows:
\begin{equation}
\underline{h}_s = \frac{1}{2T}\log \det R_{s} +
\frac{|s|}{2} \log 2\pi e
\end{equation}
Therefore if we define
\begin{equation}
g_s = \frac{1}{T}\log \det R_{s},
\label{defg}
\end{equation}
it is obvious that $g_s$ can be obtained from
$\underline{h}_s$ and vice versa. All that is involved is a
scaling of the covariance matrix $R$. Denote the vector obtained from all entries $g_s,~s\subseteq\{1,\dots,n\}$ by $g$.
For balanced inequalities there is the additional property,
\begin{lemma}
If the inequlaity $\sum_{s} \gamma_{s}H_{s}\geq 0$ is balanced then $\sum_{s} |s| \gamma_{s} = 0$.
\end{lemma}

\begin{proof}
We can simply write,
\bea
\sum_{s} |s|\gamma_{s} = \sum_{s} \sum_{i\in s} \gamma_{s} = \sum_i \left(\sum_{s: i\in s} \gamma_{s} \right) = 0
\eea
\end{proof}

Therefore the set of linear balanced information inequalities that $g$ and $\underline{h}$ satisfy is the same. Moreover any other type of inequality that $\underline{h}$ satisfies can be converted to an inequality for $g$ and vice versa and therefore the space of $g$ and $\underline{h}$ can be obtained from each other.
For simplicity, we will therefore use
$g_s$ instead of $\underline{h}_s$ throughout the
paper and use the term entropy for both $g$ and $\underline{h}$
interchangeably.

In this paper we characterize the entropy region of 3 jointly Gaussian random variables and study the minimal set of necessary and sufficient conditions for a $2^n-1$ dimensional vector to represent an entropy vector of $n$ scalar jointly Gaussian random variables for $n\geq 4$.
As equation (\ref{defg}) suggests, the entropy of any subset of random variables from a
collection of Gaussian random variables is simply the ``log'' of the principal minor of the covariance matrix corresponding to this subset. Therefore studying the entropy of Gaussian random variables involves studying the relations among principal minors of symmetric positive semi-definite matrices, i.e., covariance matrices. It has recently been noted that one of these relations is the so-called Cayley ``hyperdeterminant" \cite{hyperdet}. Therefore along the study of entropy of Gaussian random variables we also examine the hyperdeterminant relation.

The remainder of this paper is organized as follows. In the next section
we review background and some motivating results on the entropies of Gaussian random
variables. Section \ref{Gaussian_sec:main3} states the main results on the characterization of the entropy region of 3 jointly Gaussian random variables.
In Section \ref{Gaussian_sec:hyperdet} we examine the hyperdeterminant relation in connection to the entropy region of Gaussian random variables. We give a determinant formula for calculating the special $2\times 2\times 2$ hyperdeterminant. Moreover we present a new and transparent proof of the result of \cite{hyperdet} on why the principal minors of a symmetric matrix satisfy the hyperdeterminant relations.
In Section \ref{Gaussian_sec:minimal} we study the minimal set of necessary and sufficient condition for a $2^n-1$ dimensional vector to be the entropy vector of $n$ scalar jointly Gaussian random variables. For $n=4$, there are 5 such equations and we explicitly state them.


\section{Some Known Results}
\label{Gaussian_sec:known}

From (\ref{defg}) it can be easily seen that any valid
information inequality for entropies can be immediately converted into
an inequality for the (block) principal minors of a symmetric, positive
semi-definite matrix. This connection has been previously used in the literature. In fact one can study determinant inequalities by studying the corresponding entropy inequalities, see e.g. \cite{cover-determinant}.  

Let $g$ be the {\it normalized} entropy vector
corresponding to some vector-valued collection of random variables
with an $nT\times nT$ covariance matrix $R$. Further, let $m$ denote
the vector of block principal minors of $R$. Then it is clear that $m
= e^{gT}$, where the exponential acts component-wise on the entries of
$g$. Then the submodularity of entropy translates to the
following inequality for the principal minors:
\begin{equation}
m_{s \cup s'}\cdot m_{s \cap s'} \leq m_s \cdot m_{s'}
\end{equation}
In the context of determinant inequalities for a Hermitian positive semidefinite matrix this is known as the ``Koteljanskii'' inequality and is a generalization of the ``Hadamard-Fischer'' inequalities \cite{Johnson00}. Dating back at least to Hadamard in 1893, studying the determinant inequalities is an old subject which is of interest in its own right and has many applications in matrix analysis and probability theory.

Some of the interesting problems in the area of principal minor relations include characterizing the set of bounded ratios of principal minors for a given class of matrices (e.g. the class of positive definite, or the class of matrices whose all of their principal minors are positive, i.e., the P matrices) \cite{Johnson93, johnson-bddminor-new}, studying the Gaussian conditional independence structure in the context of probabilistic representations \cite{lnenicka-gauss-conditional} and detecting P matrices, e.g., via computation of all the principal minors of a given matrix \cite{GT-minor-computation}.

Although determinant inequalities have been studied extensively on their own and also through the entropy inequalities, the reverse approach of determining Gaussian entropies via the exploration of the space of principal minors has been less considered \cite{lnenicka-gauss-conditional, lnenicka-gauss}. As it turns out, this approach is deeply related to the ``principal minor assignment'' problem where a matrix with a set of fixed principal minors is sought. Recently there has been progress towards this area for symmetric matrices \cite{hyperdet, minor-assign} and we will discuss this in more detail in Sections \ref{Gaussian_sec:hyperdet} and \ref{Gaussian_sec:minimal}.

Apart from the result of \cite{lnenicka-gauss} which shows the tightness of the Zhang-Yeung non-Shannon inequality \cite{characterization} for Gaussian random variables, one of the encouraging results for studying the Gaussian random variables is that they can violate the ``Ingleton bound''. This bound is one of the best known inner bounds for $\Gamma^*_4$ \cite{characterization}.

\begin{theorem}[Ingleton inequality] \label{thm:ingleton}
\cite{ingleton} Let $v_1,\cdots,v_n$ be $n$ vector subspaces and let
$\mathcal{N}=\{1,\cdots,n\}$. Further let $s \subseteq \mathcal{N}$ and
$r_s$ be the rank function defined as the dimension of the
subspace $\oplus_{i\in s} v_i$. Then for any subsets
$s_1,s_2,s_3,s_4 \subseteq \mathcal{N}$, we have
\vspace*{-.15in}
\begin{eqnarray}
& & r_{s_1}+r_{s_2}+r_{s_1\cup s_2\cup s_3}+
r_{s_1\cup s_2\cup s_4}+r_{s_3\cup s_4}
\nonumber \\
&& \hspace*{-.25in}
-r_{s_1\cup s_2}-r_{s_1\cup s_3}-r_{s_1\cup s_4}
-r_{s_2\cup s_3}-r_{s_2\cup s_4}\leq 0
\end{eqnarray}
\end{theorem}

Ingleton inequality was first obtained
for the rank of vector spaces. However it turns out that certain types of
entropy functions, in particular all linear representable
(corresponding to linear codes over finite fields) and
pseudo-abelian group characterizable entropy functions also satisfy
this inequality and hence fall into this inner bound
\cite{chan-abelian, chan-group}. However if we consider 4
jointly Gaussian random variables, we interestingly find that they can
violate the Ingleton bound.
Consider the following covariance matrix:
\begin{equation}
\ba{cccc} 1 & \varepsilon & a & a \\
\varepsilon & 1 & a & a\\
a & a & 1 & 0\\
a & a & 0 & 1
\ea
\end{equation}
To violate the Ingleton inequality we need to have:
\begin{eqnarray}
\nonumber && g_1 + g_2 + g_{123} + g_{124} + g_{34}\\
&& \hspace{10pt} - g_{12} - g_{13} - g_{14} - g_{23} -g_{24} \geq 0
\end{eqnarray}
or equivalently in terms of the minors $m$:
\bea
\frac{m_1m_2m_{123}m_{124}m_{34}}{m_{12}m_{13}m_{14}m_{23}m_{24}} \geq 1
\eea
Substituting for values of $m$ from the covariance matrix and simplifying we obtain:
\begin{eqnarray} \label{eqn:ingletonvio}
\frac{1-\varepsilon}{1+\varepsilon} \geq \left(\frac{1-2a^2+a^4}{1-2a^2+\varepsilon}\right)^2
\end{eqnarray}
Moreover imposing positivity conditions for this matrix to correspond to a true covariance matrix gives \mbox{$0\leq a^2 \leq 0.5$}, \mbox{$4a^2-1 \leq \varepsilon \leq 1$}.
Solving inequality (\ref{eqn:ingletonvio}) subject to these constraints yields a region of permissible $\varepsilon$ and $a^2$ (Fig \ref{Gaussian_fig:Gvio}). In particular the point $\varepsilon = 0.25,~a=0.5$ lies in this
region. Interestingly enough, this example has also been discovered in the context of determinantal inequalities in
\cite{Johnson93}. 
\begin{figure}[t]
\centering
\scalebox{0.6}{\includegraphics{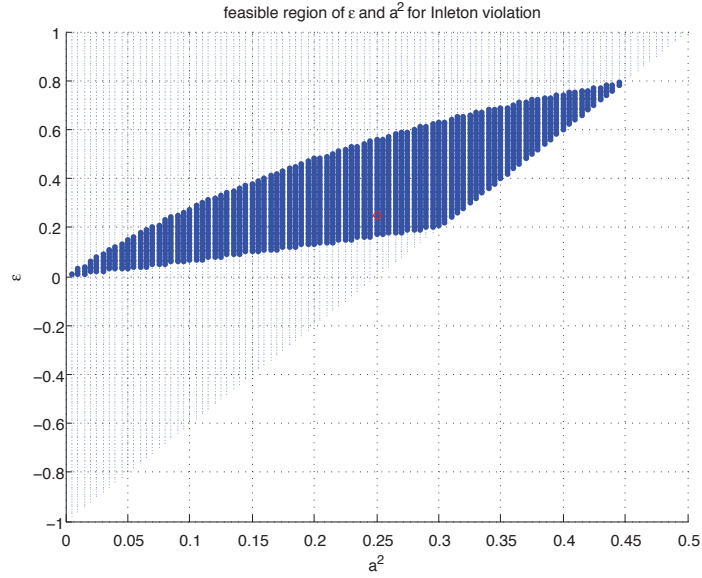}}
\caption{Feasible region of $\varepsilon$ and $a^2$ for the specific Ingleton violating example}
\label{Gaussian_fig:Gvio}
\end{figure}

Taking these results into account, we will hence study the Gaussian entropy region for 2 and 3 random variables and give the minimal number of necessary and sufficient conditions for a $2^n-1$ dimensional vector to correspond to the entropy of $n$ scalar jointly Gaussian random variables in the following sections.

\section{Entropy Region of 2 and 3 Gaussian Random Variables}
\label{Gaussian_sec:main3}

The mentioned results in the previous section (violation of the Ingleton bound and tightness of the
non-Shannon inequality) lead one to speculate whether the entropy
region for arbitrary continuous random variables is equal to 
the entropy region of (vector-valued) Gaussian ones. Although this is the case for $n=2$ random variables, it is not true for $n=3$. What is true for $n=3$ is that the entropy region of 3 arbitrary continuous random variables can be obtained from the convex cone of the region of 3 scalar-valued Gaussian random variables. 

\subsection{$n=2$}
Entropy region of 2 jointly Gaussian random variables is trivially equal to the whole entropy region of 2 arbitrary distributed continuous random variables.
\begin{theorem}
The entropy region of 2 jointly Gaussian random variables is described by the single inequality $g_{12}\leq g_1+g_2$ and is equal to the entropy region of 2 arbitrary distributed continuous random variables.
\end{theorem}
\begin{proof}
Since it is known that the continuous entropy region is described by the single balanced inequality $h_{12}\leq h_1 + h_2$, to prove the theorem it is sufficient to show that any entropy vector $[h_1,h_2,h_{12}]$ satisfying this inequality may be described by 2 jointly Gaussians and this is trivial to show. 
\end{proof}
 
\subsection{Main Results for $n=3$}
Although we consider vector-valued jointly Gaussian random variables, for $n=3$ we interestingly find that considering the convex hull of {\it scalar} jointly Gaussian random variables is sufficient for characterizing the Gaussian entropy region.

\begin{theorem} \label{thm:vector-vs-scalar}
The entropy region of 3 {\it vector-valued} Gaussian random variables can be obtained from the convex hull of {\it scalar} Gaussian random variables.
\end{theorem}
 
The main result about the convex cone of the entropy region of 3 jointly Gaussian random variables is formalized in the next theorem:

\begin{theorem}[Convex Cone of the Entropy Region of 3 Scalar-Valued Gaussian Random Variables] 
~The convex cone generated by the entropy region of 3 scalar-valued Gaussian random variables gives the entropy region of 3 arbitrary continuously distributed random variables. 


\label{thm:cone}
\end{theorem}

This theorem states that one can indeed construct the entropy region of $n=3$ continuous random variables
from the entropy region of Gaussian random variables and therefore it encourages the study of
Gaussians for $n\geq 4$.
Moreover Theorem \ref{thm:cone} addresses the ``convex cone'' of the Gaussian entropy region and as it turns out for most practical purposes characterizing the ``convex cone'' is sufficient.
The problem of characterizing the entropy region of Gaussian random variables itself, rather than its convex cone, is more complicated. For 3 random variables we state the following conjecture:

\begin{conjecture}[Entropy Region of 3 Vector-Valued Jointly Gaussian Random Variables] 
Let vector $g$ defined as \mbox{$g = [ g_1, g_2, g_3, g_{12}, g_{23}, g_{31}, g_{123}]^t$} be an entropy vector generated by 3 vector-valued Gaussian random variables. Define $x_k = e^{g_{ij}-g_i-g_j}$
  and $\tilde{y} = \frac{\prod_kx_k}{\max_kx_k}+2\max_kx_k-\sum_kx_k$. The closure of the Gaussian entropy region generated by such $g$ vectors is characterized by,
\begin{enumerate}
\item For $\tilde{y}\leq 0$:
\be
g_{ij}\leq g_i+g_j~~,~~g_{123}\leq\min_j(g_{ij}+g_{jk}-g_j).
\label{3ent}
\ee
\item For $\tilde{y}>0$:
\begin{eqnarray}
\hspace*{-.35in} g_{ij}\leq g_i+g_j~~,~~g_{123}\leq \sum_kg_k+\log\left(\max\left[0,-2+\sum_k x_k+2\sqrt{\prod_k(1-x_k)}\right]\right)
\label{tighter123}
\end{eqnarray}
\end{enumerate}
\label{thm:3Gauss}

\end{conjecture}

In other words we conjecture that the entropy region for three Gaussian random variables is simply given by the
above inequalities. Thus, when $\tilde{y}\leq 0$, the Gaussian
entropy region coincides with the continuous entropy region; however,
when $\tilde{y}>0$ (and this can happen for some valid entropy vectors), we have
the tighter upper bound (\ref{tighter123}) on
$g_{123}$. In other words the actual Gaussian entropy region for $n=3$ vector-valued
random variables is {\em strictly} smaller than the entropy
region of 3 arbitrarily distributed continuous random variables.

We strongly believe the above conjecture to be true. The missing gap in our proof is a certain function inequality, which all our simulations suggest to be true (see Conjecture \ref{conj:func-conj}).

\subsection{Proof of Main Results for $n=3$}
\label{sec:proofs}

In what follows we give the proof of the results stated in the previous section for $n=3$. 
The basic idea
is to determine the structure of the Gaussian random variables that
generate the {\em boundary} of the entropy region for Gaussians, and
then to determine what the boundary entropies are. We need a few lemmas:

%

\begin{lemma}[Boundary of the Gaussian Entropy Region] The boundary of
  the Gaussian entropy region is generated by the concatenation of a
  set of vector valued Gaussian random variables with covariance
\be
\ba{ccc} \alpha_{11}I_{\hat T} & \alpha_{12}\Phi_{12} &
  \alpha_{13}\Phi_{13} \\
\alpha_{12}\Phi_{12}^{t} & \alpha_{22}I_{\hat T} & \alpha_{23}\Phi_{23} \\
\alpha_{13}\Phi_{13}^{t} & \alpha_{23}\Phi_{23}^{t} & \alpha_{33}I_{\hat T}\ea ,
\ee
where the $\Phi_{ij}$ are {\em orthogonal} matrices such that $\Phi_{13}^t\Phi_{12}\Phi_{23} = \mbox{sign}(\alpha_{12}\alpha_{13}\alpha_{23}) I$, and another set
  of independent vector-valued Gaussian random variables with covariance
\be
\ba{ccc} \alpha_{11}I_{T-{\hat T}} & 0 & 0 \\
0 & \alpha_{22}I_{T-{\hat T}} & 0 \\
0 & 0 & \alpha_{33}I_{T-{\hat T}} \ea .
\ee
\label{lem:boundary}
\end{lemma}

\begin{proof}
To find the boundary region for 3 jointly Gaussian random variables, we can maximize linear functions of the entropy vector. We can therefore take an arbitrary set of constants $\gamma_s, s\subseteq \{1,2,3\}$ and  
solve the following maximization problem:
\be
\max_{h} \sum_{s\subseteq \{1,2,3\}} \gamma_{s} h_s \label{Gaussian_eqn:hoptimization}
\ee
or equivalently,
\begin{equation}\label{Gaussian_eqn:optimization}
\max_{\tilde{R}}\sum_{s\subseteq \{1,2,3\}}{\gamma_s \log\det{\tilde{R}_{[s]} }},
\end{equation}
where $\tilde{R}$ is the $3T\times3T$ block covariance matrix and $\tilde{R}_{[s]}$ denotes the submatrix of $\tilde{R}$ whose rows and columns are indexed by $s$.
We shall assume all the principal minors of $\tilde{R}$ are nonzero. The optimization problems (\ref{Gaussian_eqn:hoptimization}-\ref{Gaussian_eqn:optimization}) also come about when we fix any 6 of the entropies and try to maximize the last one. 
KKT conditions necessitate that the derivative of (\ref{Gaussian_eqn:optimization}) with respect to $\tilde{R}$ be zero, i.e. \mbox{$\frac{\partial}{\partial \tilde{R}} \left(\sum_{s\subseteq \{1,2,3\}}{\gamma_s \log\det{\tilde{R}_{[s]}}}\right) = 0$}. To compute the derivatives we note that \mbox{$\frac{\partial}{\partial X} \log \det X = X^{-t} $}. However since covariance matrix $\tilde{R}$ is symmetric, we can further write \mbox{$\frac{\partial}{\partial X} \log \det X = X^{-1} $}.
If we adopt the following notation,
\bea
\nonumber && \hspace{-20pt} \tilde{S} = \left(\begin{array}{cc} \tilde{S}_{11} & \tilde{S}_{12} \\ \tilde{S}_{21} & \tilde{S}_{22} \end{array}\right)  = \left(\begin{array}{cc} \tilde{R}_{11} & \tilde{R}_{12} \\ \tilde{R}_{21} & \tilde{R}_{22} \end{array}\right)^{-1},  \tilde{W} = \left(\begin{array}{cc} \tilde{W}_{11} & \tilde{W}_{13} \\ \tilde{W}_{31} & \tilde{W}_{33} \end{array}\right)  = \left(\begin{array}{cc} \tilde{R}_{11} & \tilde{R}_{13} \\ \tilde{R}_{31} & \tilde{R}_{33} \end{array}\right)^{-1}\\
&& \hspace{-20pt} \tilde{U} = \left(\begin{array}{cc} \tilde{U}_{22} & \tilde{U}_{23} \\ \tilde{U}_{32} & \tilde{U}_{33} \end{array}\right)  = \left(\begin{array}{cc} \tilde{R}_{22} & \tilde{R}_{23} \\ \tilde{R}_{32} & \tilde{R}_{33} \end{array}\right)^{-1}, ~ \tilde{V}_{ij} = (\tilde{R}^{-1})_{ij} \label{Gaussian_eqn:studef1}
\eea
Then we obtain, 
\footnote{Note that had we taken the derivative with respect to a symmetric matrix $X$ from the beginning, then we had, \mbox{$\frac{\partial}{\partial X} \log \det X = 2X^{-1} - \mbox{diag}(X^{-1}) $}, where ``\mbox{diag($\cdot$)}'' denotes the diagonal elements of its argument. However this derivation would also result in equation (\ref{Gaussian_eqn:optimMatrix}).}
\bea
\nonumber  && \gamma_{1} \hspace{0pt} \left( \begin{array}{ccc} \tilde{R}_{11}^{-1} & 0 & 0 \\ 0 & 0 & 0 \\ 0 & 0 & 0 \end{array}\right) + \gamma_{2} \left(\begin{array}{ccc} 0 & 0 & 0 \\ 0 & \tilde{R}_{22}^{-1} & 0 \\ 0 & 0 & 0 \end{array}\right) + \gamma_3 \left(\begin{array}{ccc} 0 & 0 & 0 \\ 0 & 0 & 0 \\ 0 & 0 & \tilde{R}_{33}^{-1} \end{array}\right) +
\gamma_{12} \left(\begin{array}{ccc} \tilde{S}_{11} & \tilde{S}_{12} & 0 \\ \tilde{S}_{21} & \tilde{S}_{22} & 0 \\ 0 & 0 & 0 \end{array}\right) \\
&& \hspace{20pt} +\gamma_{13} \left(\begin{array}{ccc} \tilde{W}_{11} & 0 & \tilde{W}_{13} \\ 0 & 0 & 0 \\ \tilde{W}_{31} & 0 & \tilde{W}_{33} \end{array} \right) 
 +  \gamma_{23} \hspace{0pt} \left(\begin{array}{ccc} 0 & 0 & 0 \\ 0 & \tilde{U}_{22} & \tilde{U}_{23}\\ 0 & \tilde{U}_{32} & \tilde{U}_{33} \end{array}\right)
+  \gamma_{123} \tilde{R}^{-1} = 0 \label{Gaussian_eqn:optimMatrix}
\eea
Now if we assume $\det \tilde{R}_{ii} = \alpha_{ii}^T, \alpha_{ii}>0$, then we can define the following ``unit-determinant'' matrix:
\bea
L = \left(\begin{array}{ccc}\frac{1}{\sqrt{\alpha_{11}}}\tilde{R}_{11}^{1/2} & 0 & 0 \\ 0 & \frac{1}{\sqrt{\alpha_{22}}}\tilde{R}_{22}^{1/2} & 0 \\ 0 & 0 & \frac{1}{\sqrt{\alpha_{33}}}\tilde{R}_{33}^{1/2} \end{array}\right)
\eea
Let $L_{[s]}, s\subseteq \{1,2,3\}$ be the submatrix of $L$ obtained from choosing rows and columns of $L$ that are indexed by $s$. 
Then if we multiply (\ref{Gaussian_eqn:optimMatrix}) from left and right by $L$ we obtain:
\bea
\nonumber  && \hspace{-20pt} \gamma_{1} \hspace{0pt} \left( \begin{array}{ccc} L_{[1]}\tilde{R}_{11}^{-1}L_{[1]} & 0 & 0 \\ 0 & 0 & 0 \\ 0 & 0 & 0 \end{array}\right) 
+ \gamma_{2} \left(\begin{array}{ccc} 0 & 0 & 0 \\ 0 & L_{[2]}\tilde{R}_{22}^{-1}L_{[2]} & 0 \\ 0 & 0 & 0 \end{array}\right) 
+ \gamma_3 \left(\begin{array}{ccc} 0 & 0 & 0 \\ 0 & 0 & 0 \\ 0 & 0 & L_{[3]}\tilde{R}_{33}^{-1}L_{[3]} \end{array}\right) \\
&& \hspace{-20pt} +\gamma_{12} \left(\begin{array}{cc} 
L_{[1,2]}\left(\begin{array}{cc} \tilde{S}_{11} & \tilde{S}_{12} \\ \tilde{S}_{21} & \tilde{S}_{22} \end{array}\right)L_{[1,2]} & 
\begin{array}{c} 0 \\ 0\end{array}\\ 
\begin{array}{cc} 0 & \hspace{40pt} 0\end{array} & 0 
 \end{array}\right) 
+\gamma_{13} \left(\begin{array}{ccc} \left( L_{[1,3]}\tilde{W}L_{[1,3]} \right)_{1,1} & 0 & \left( L_{[1,3]}\tilde{W}L_{[1,3]} \right)_{1,2} \\ 
0 & 0 & 0 \\ \left( L_{[1,3]}\tilde{W}L_{[1,3]} \right)_{2,1}  & 0 &  \left( L_{[1,3]}\tilde{W}L_{[1,3]} \right)_{2,2}
\end{array} \right) \\
&& \hspace{-20pt} +  \gamma_{23} \hspace{0pt} \left(\begin{array}{cc} 
0 &  \begin{array}{cc} 0 & \hspace{40pt} 0\end{array}  \\ 
\begin{array}{c} 0 \\ 0\end{array}  & L_{[2,3]}\left( \begin{array}{cc} \tilde{U}_{22} & \tilde{U}_{23}\\ \tilde{U}_{32} & \tilde{U}_{33} \end{array}\right)L_{[2,3]}
\end{array}\right)
+  \gamma_{123} L\tilde{R}^{-1}L = 0 
\eea
Due to the special structure of $L$, we have $(L_{[s]})^{-1}R_{[s]}(L_{[s]})^{-1} = (L^{-1}R L^{-1})_{[s]}$
therefore if we define
\bea
\label{eqn:similarityRLR}
R = L^{-1} \tilde{R} L^{-1}
\eea
and $S, W, U$ and $V$ also similar to (\ref{Gaussian_eqn:studef1}), i.e.,
\bea
\nonumber && \hspace{-20pt} S = \left(\begin{array}{cc} {S}_{11} & {S}_{12} \\ {S}_{21} & {S}_{22} \end{array}\right)  = \left(\begin{array}{cc} {R}_{11} & {R}_{12} \\ {R}_{21} & {R}_{22} \end{array}\right)^{-1},  W = \left(\begin{array}{cc} {W}_{11} & {W}_{13} \\ {W}_{31} & {W}_{33} \end{array}\right)  = \left(\begin{array}{cc} {R}_{11} & {R}_{13} \\ {R}_{31} & {R}_{33} \end{array}\right)^{-1}\\
&& \hspace{-20pt} U = \left(\begin{array}{cc} {U}_{22} & {U}_{23} \\ {U}_{32} & {U}_{33} \end{array}\right)  = \left(\begin{array}{cc} {R}_{22} & {R}_{23} \\ {R}_{32} & {R}_{33} \end{array}\right)^{-1}, ~ {V}_{ij} = ({R}^{-1})_{ij} \label{Gaussian_eqn:studef2}
\eea
it follows that (\ref{Gaussian_eqn:optimMatrix}) will be satisfied by $R, S, W, U, V$ instead of $\tilde{R}, \tilde{S}, \tilde{W}, \tilde{U}, \tilde{V}$. In other words we have the following equation:
\bea
\nonumber  && \hspace{-30pt} \left( \begin{array}{ccc} \gamma_1 R_{11}^{-1} & 0 & 0 \\ 0 & \gamma_2 R_{22}^{-1} & 0 \\ 0 & 0 & \gamma_3 R_{33}^{-1} \end{array}\right) + 
\gamma_{12} \left(\begin{array}{ccc} S_{11} & {S}_{12} & 0 \\ S_{21} & S_{22} & 0 \\ 0 & 0 & 0 \end{array}\right) 
 +\gamma_{13} \left(\begin{array}{ccc} W_{11} & 0 & W_{13} \\ 0 & 0 & 0 \\ W_{31} & 0 & W_{33} \end{array} \right) \\
&& \hspace{-20pt} +  \gamma_{23} \hspace{0pt} \left(\begin{array}{ccc} 0 & 0 & 0 \\ 0 & U_{22} & U_{23}\\ 0 & U_{32} & U_{33} \end{array}\right)
+  \gamma_{123} R^{-1} = 0 
\label{Gaussian_eqn:optimMatrix_forR}
\eea
Multiplying (\ref{Gaussian_eqn:optimMatrix_forR}) by $R$ from the right we obtain,
\bea
\nonumber && \hspace{-17pt} \left(\begin{array}{ccc} \gamma_1 I & \gamma_1 {R}_{11}^{-1}{R}_{12} & \gamma_1 {R}_{11}^{-1}{R}_{13} \\ \gamma_2{R}_{22}^{-1}{R}_{21} & \gamma_2I & \gamma_2{R}_{22}^{-1}{R}_{23} \\ \gamma_3{R}_{33}^{-1}{R}_{31} & \gamma_3{R}_{33}^{-1}{R}_{32} & \gamma_3I \end{array} \right)
+  \gamma_{12} \left(\begin{array}{cc} \begin{array}{cc} I & 0 \\ 0 & I \end{array} & \hspace{-10pt} \left( \hspace{-3pt} \begin{array}{cc} {S}_{11} & \hspace{-2pt} {S}_{12} \\ {S}_{21} & \hspace{-2pt} {S}_{22} \end{array}\hspace{-3pt} \right) \hspace{-4pt} \left( \hspace{-3pt} \begin{array}{c} {R}_{13}\\ {R}_{23} \end{array} \hspace{-3pt} \right) \\ \begin{array}{cc} 0 & 0 \end{array} & \hspace{-10pt} 0 \end{array}\right)\\
 && +  \gamma_{13}  \hspace{0pt} \left(\begin{array}{ccc} I & {W}_{11}{R}_{12}+{W}_{13}{R}_{32} & 0 \\ 0 & 0 & 0 \\ 0 & {W}_{31}{R}_{12} + {W}_{33}{R}_{32} & I \end{array}\right)
+  \gamma_{23} \left(\begin{array}{cc} 0 & \hspace{-10pt} \begin{array}{cc} 0 & 0 \end{array} \\  \left( \hspace{-3pt} \begin{array}{cc} {U}_{22} & \hspace{-2pt} {U}_{23} \\ {U}_{32} & \hspace{-2pt} {U}_{33} \end{array} \hspace{-3pt} \right) \hspace{-4pt} \left( \hspace{-3pt} \begin{array}{c} {R}_{21} \\ {R}_{31} \end{array} \hspace{-3pt} \right) & \hspace{-10pt} \begin{array}{cc} I & 0 \\ 0 & I \end{array}  \end{array}\right) + \gamma_{123}I = 0 
\label{Gaussian_eqn:optimMatrix2}
\eea
Note that equating the diagonal elements to zero results in the following constraints on $\gamma$ coefficients:
\bea
\label{balance1} \gamma_1+\gamma_{12}+\gamma_{13}+\gamma_{123} = 0\\
\label{balance2} \gamma_2+\gamma_{12}+\gamma_{23}+\gamma_{123} = 0\\
\label{balance3} \gamma_3+\gamma_{13}+\gamma_{23}+\gamma_{123} = 0
\eea
These imply that the equations representing the touching hyperplanes to the Gaussian region should be balanced (see Theorem \ref{thm:contin-ineq}).
Now considering blocks (2,1), (3,1) together, (1,2), (3,2) with each other and (1,3), (2,3) simultaneously in (\ref{Gaussian_eqn:optimMatrix2}) and noting that $R_{ii} = \alpha_{ii}I$,  we obtain 

\begin{eqnarray} \label{Rij-1}
\left(\ba{cc} \frac{\gamma_1}{\alpha_{11}}I  & 0 \\ 0 & \frac{\gamma_2}{\alpha_{22}}I \ea  +
 \gamma_{12} {\ba{cc} R_{11} & R_{12} \\ R_{21} & R_{22} \ea}^{-1} \right) \ba{c} R_{13}\\ R_{23} \ea = 0
\end{eqnarray}

\begin{eqnarray} \label{Rij-2}
\left(\ba{cc} \frac{\gamma_1}{\alpha_{11}}I  & 0 \\ 0 & \frac{\gamma_3}{\alpha_{33}}I \ea  +
 \gamma_{13} {\ba{cc} R_{11} & R_{13} \\ R_{31} & R_{33} \ea}^{-1} \right) \ba{c} R_{12}\\ R_{32} \ea = 0
\end{eqnarray}

\begin{eqnarray} \label{Rij-3}
\left(\ba{cc} \frac{\gamma_2}{\alpha_{22}}I  & 0 \\ 0 & \frac{\gamma_3}{\alpha_{33}}I \ea  +
 \gamma_{23} {\ba{cc} R_{22} & R_{23} \\ R_{32} & R_{33} \ea}^{-1} \right) \ba{c} R_{21}\\ R_{31} \ea = 0
\end{eqnarray}
Simplifying equations (\ref{Rij-1})--(\ref{Rij-3}) by multiplying each by the relevant ${\ba{cc} R_{ii} & R_{ij} \\ R_{ji} & R_{jj} \ea}$, we obtain:
\begin{eqnarray} \label{RijCond-1}
\hspace{-35pt} &&\ba{cc} (\gamma_1+\gamma_{12} )I & \frac{\gamma_2}{\alpha_{22}}R_{12} \\ \frac{\gamma_1}{\alpha_{11}}R_{21} & (\gamma_2+\gamma_{12} )I \ea  \ba{c} R_{13}\\ R_{23} \ea = 0
\end{eqnarray}

\begin{eqnarray} \label{RijCond-2}
\hspace{-35pt} &&\ba{cc} (\gamma_1+\gamma_{13} )I & \frac{\gamma_3}{\alpha_{33}}R_{13} \\ \frac{\gamma_1}{\alpha_{11}}R_{31} & (\gamma_3+\gamma_{13} )I \ea  \ba{c} R_{12}\\ R_{32} \ea = 0
\end{eqnarray}

\begin{eqnarray} \label{RijCond-3}
\hspace{-35pt} &&\ba{cc} (\gamma_2+\gamma_{23} )I & \frac{\gamma_3}{\alpha_{33}}R_{23} \\ \frac{\gamma_2}{\alpha_{22}}R_{32} & (\gamma_3+\gamma_{23} )I \ea  \ba{c} R_{21}\\ R_{31} \ea = 0
\end{eqnarray}
Now if the $2T\times T$ matrix $\ba{cc} R_{ik}^t & R_{jk}^t \ea^{t},~i,j,k \in \{1,2,3\}$ were full rank, the rank of the left $2T\times 2T$ matrix in either of the equations (\ref{RijCond-1})--(\ref{RijCond-3}) would be $T$ and therefore the Schur complement of its (1,1) block should be zero, i.e.:
\begin{equation}
(\gamma_j+\gamma_{ij})I-\frac{\gamma_i\gamma_j}{\alpha_{ii}\alpha_{jj}(\gamma_i+\gamma_{ij}) }R_{ji}R_{ij} = 0
\end{equation}
in other words:
\begin{equation} \label{unitary_cond}
R_{ji}R_{ij}= R_{ij}R_{ji} = \left( \frac{(\gamma_i+\gamma_{ij})(\gamma_j+\gamma_{ij}) }{\gamma_i\gamma_j } \alpha_{ii}\alpha_{jj} \right)  I
\end{equation}
Since $R$ is symmetric, $R_{ji} = R_{ij}^t$. This implies that off-diagonal blocks of $R$ are multiples of an orthogonal matrix, i.e., 
\be
R_{ij} = \alpha_{ij} \Phi_{ij}
\label{Rsym}
\ee
for some orthogonal matrix $\Phi_{ij}$ and $\alpha_{ij}$ such that
\be
\alpha_{ij}^2 = {\frac{(\gamma_i+\gamma_{ij})(\gamma_j+\gamma_{ij}) }{\gamma_i\gamma_j } \alpha_{ii}\alpha_{jj}}
\label{alphaij}
\ee
Stating (\ref{Rsym}) explicitly, we have:
\bea
\label{R12} R_{12} = R_{21}^2 = \alpha_{12} \Phi_{12}\\
\label{R13} R_{13} = R_{31}^2 = \alpha_{13} \Phi_{13}\\
\label{R23} R_{23} = R_{32}^2 = \alpha_{23} \Phi_{23}
\eea
Replacing for $R_{ij}$ from (\ref{R12})--(\ref{R23}) in equations (\ref{RijCond-1})--(\ref{RijCond-3}) we obtain 6 equations, which turn out to be all the same as the following equation when we consider the balancedness constraints of (\ref{balance1})--(\ref{balance3}) and definition of $\alpha_{ij}^2$ in (\ref{alphaij}):
\be
\Phi_{13} = -\frac{\gamma_2}{\gamma_1+\gamma_{12} } ~ \frac{\alpha_{12}\alpha_{23} }{\alpha_{2}\alpha_{13} }~ \Phi_{12}\Phi_{23}
\label{phi13phi12phi23}
\ee
Simplifying (\ref{phi13phi12phi23}) using the fact that $\alpha_{ij} = \pm \sqrt{\frac{(\gamma_i+\gamma_{ij})(\gamma_j+\gamma_{ij}) }{\gamma_i\gamma_j } \alpha_{ii}\alpha_{jj} } $ from (\ref{alphaij}), we obtain that:
\bea
\Phi_{13} = \begin{cases}
+\Phi_{12}\Phi_{23}~~~~\mbox{if}~~~\alpha_{12}\alpha_{13}\alpha_{23}>0 \\ -\Phi_{12}\Phi_{23}~~~~\mbox{if}~~~\alpha_{12}\alpha_{13}\alpha_{23}<0
\end{cases}
\eea

Now note that in the general case $\ba{cc} R_{ik}^t & R_{jk}^t \ea^{t} $ in equations (\ref{RijCond-1})--(\ref{RijCond-3}) need not be full rank. Therefore there is a $T\times T$ unitary matrix $\theta_{ij}$ such that:
\begin{equation} \label{fullrank_trans}
\ba{c} R_{ik}\\ R_{jk} \ea \theta_{ij} = \ba{cc} \overline{R}_{ik} & 0 \\ \overline{R}_{jk} & 0 \ea
\end{equation}
Writing explicitly:
\bea
\ba{c} R_{13}\\ R_{23} \ea \theta_{12} = \ba{cc} \overline{R}_{13} & 0 \\ \overline{R}_{23} & 0 \ea,~~~
\ba{c} R_{12}\\ R_{32} \ea \theta_{13} = \ba{cc} \overline{R}_{12} & 0 \\ \overline{R}_{32} & 0 \ea,~~~
\ba{c} R_{21}\\ R_{31} \ea \theta_{23} = \ba{cc} \overline{R}_{21} & 0 \\ \overline{R}_{31} & 0 \ea,
\eea
where $\ba{cc} \overline{R}_{ik}^t & \overline{R}_{jk}^t \ea^t$ is now full rank and we can assume its column rank (as well as its column size) to be $T_{ij}$ where $T_{ij}<T$.
This suggests doing a similarity transformation on $R$ with the following unitary matrix without affecting the block principal minors:
\begin{equation}
\Theta = \ba{ccc} \theta_{23} & 0 & 0\\0 & \theta_{13} & 0 \\ 0 & 0 & \theta_{12} \ea
\end{equation}
From which we obtain:
\begin{equation}
\Theta^{*}R\Theta = \ba{ccc} \alpha_{11} I & \theta^*_{23} R_{12} \theta_{13} & \theta^*_{23} R_{13} \theta_{12} \\
\theta^*_{13} R_{21} \theta_{23} & \alpha_{22} I & \theta^*_{13} R_{23} \theta_{12} \\
\theta^*_{12} R_{31} \theta_{23} & \theta^*_{12} R_{32} \theta_{31} & \alpha_{33} I \ea
\end{equation}
Considering $R_{21} \theta_{23}$ and $\theta^*_{31} R_{21}$ simultaneously and using (\ref{fullrank_trans}) we have,
\bea
&& R_{21} \theta_{23} = \left(\begin{array}{cc}\overline{R}_{21} & 0 \end{array}\right)\\
&& \theta_{13}^* R_{21} = (R_{12}\theta_{13})^* = \left(\begin{array}{c}\overline{R}_{12}^* \\ 0 \end{array}\right)
\eea
Therefore we can simply obtain the following structure for $\theta^*_{13} R_{21} \theta_{23}$:
\begin{equation}
\theta^*_{13} R_{21} \theta_{23} = \ba{cc} \hat{R}_{21} & 0 \\ 0 & 0 \ea
\end{equation}
where the dimension of $\hat{R}_{21}$ is $T_{13}\times T_{23}$. A similar argument for other elements yields the following structure for $\Theta^{*}R\Theta$:
\begin{equation} \label{Rstructure}
\ba{cccccc}\hspace{-5pt} \alpha_{11} I_{T_{23}} & \hspace{-10pt} 0 & \hspace{-10pt} \hat{R}_{12} & \hspace{-10pt} 0 & \hspace{-10pt} \hat{R}_{13} & \hspace{-10pt} 0 \\
\hspace{-10pt} 0 & \hspace{-10pt} \alpha_{11} I_{T-T_{23}} & \hspace{-10pt} 0 & \hspace{-10pt} 0 & \hspace{-10pt} 0 & \hspace{-10pt} 0 \\
\hspace{-10pt} \hat{R}_{21} & \hspace{-10pt} 0 & \hspace{-10pt} \alpha_{22} I_{T_{13}} & \hspace{-10pt} 0 & \hspace{-10pt} \hat{R}_{23} & \hspace{-10pt} 0 \\
\hspace{-10pt} 0 & \hspace{-10pt} 0 & \hspace{-10pt} 0 & \hspace{-10pt} \alpha_{22} I_{T-T_{13}} & \hspace{-10pt} 0 & \hspace{-10pt} 0 \\
\hspace{-10pt} \hat{R}_{31} & \hspace{-10pt} 0 & \hspace{-10pt} \hat{R}_{32} & \hspace{-10pt} 0 & \hspace{-10pt} \alpha_{33} I_{T_{12}} & \hspace{-10pt} 0 \\
\hspace{-10pt} 0 & \hspace{-10pt} 0 & \hspace{-10pt} 0  & \hspace{-10pt} 0 & \hspace{-10pt} 0 & \hspace{-10pt} \alpha_{33} I_{T-T_{12}} \hspace{-5pt} \ea
\end{equation}
Now define $R' = \Theta^* R \Theta$ and similar to (\ref{Gaussian_eqn:studef2}) let $S' = (R'_{[1,2]}) ^{-1}, W' = (R'_{[1,3]})^{-1}, U' = (R'_{[2,3]})^{-1}$ and $V' = R'^{-1}$ where for $s\subseteq \{1,2,3\}$, $R'_{[s]}$ denotes the submatrix of $R'$ obtained from choosing the block rows and block columns indexed by $s$.
If we multiply (\ref{Gaussian_eqn:optimMatrix_forR}) from the left by $\Theta^*$ and from the right by $\Theta$, then it turns out that (\ref{Gaussian_eqn:optimMatrix_forR}) is satisfied when $R, S, W, U, V$ are replaced $R', S', W', U', V'$.  Therefore we can write equations (\ref{RijCond-1})--(\ref{RijCond-3}) for $R'$. Doing so gives the following:
\begin{eqnarray} 
\hspace{-35pt} &&\ba{cc} (\gamma_i+\gamma_{ij} )I & \frac{\gamma_j}{\alpha_{jj}}R'_{ij} \\ \frac{\gamma_i}{\alpha_{ii}}R'_{ji} & (\gamma_j+\gamma_{ij} )I \ea  \ba{c} R'_{ik}\\ R'_{jk} \ea = 0
\end{eqnarray}
which if we replace for values of $R'_{ij}$ from (\ref{Rstructure}) and assume that $T_{ij}$ and $T_{ji}$ represent the same value, we obtain:
\bea
\left[ \begin{array}{cccc} (\gamma_i+\gamma_{ij})I_{T_{jk}} & 0 & \frac{\gamma_j}{\alpha_{jj}} \hat{R}_{ij} & 0\\
0 & I_{T-T_{jk}} & 0 & 0\\
\frac{\gamma_i}{\alpha_{ii}}\hat{R}_{ji} & 0 & (\gamma_j+\gamma_{ij})I_{T_{ik}} & 0\\
0 & 0 & 0 & (\gamma_j+\gamma_{ij})I_{T-T_{ik}}
\end{array}\right]
\left[\begin{array}{cc} \hat{R}_{ik} & 0\\ 0 & 0\\ \hat{R}_{jk} & 0\\ 0 & 0
\end{array}\right]
= 0
\eea
From which it follows that:
\begin{equation} \label{refinedRij}
\ba{cc} (\gamma_i+\gamma_{ij})I_{T_{jk}} & \frac{\gamma_j}{\alpha_{jj}}\hat{R}_{ij} \\ \frac{\gamma_i}{\alpha_{ii}}\hat{R}_{ji} & (\gamma_j+\gamma_{ij})I_{T_{ik}} \ea  \ba{c} \hat{R}_{ik}\\ \hat{R}_{jk} \ea = 0
\end{equation}
Stating (\ref{refinedRij}) explicitly, we have the following 3 equations:
\bea
\ba{cc} (\gamma_1+\gamma_{12})I_{T_{23}} & \frac{\gamma_2}{\alpha_{22}}\hat{R}_{12} \\ \frac{\gamma_1}{\alpha_{11}}\hat{R}_{21} & (\gamma_2+\gamma_{12})I_{T_{13}} \ea  \ba{c} \hat{R}_{13}\\ \hat{R}_{23} \ea = 0
\label{fr12}
\eea

\bea
\ba{cc} (\gamma_1+\gamma_{13})I_{T_{23}} & \frac{\gamma_3}{\alpha_{33}}\hat{R}_{13} \\ \frac{\gamma_1}{\alpha_{11}}\hat{R}_{31} & (\gamma_3+\gamma_{13})I_{T_{12}} \ea  \ba{c} \hat{R}_{12}\\ \hat{R}_{32} \ea = 0
\label{fr13}
\eea

\bea
\ba{cc} (\gamma_2+\gamma_{23})I_{T_{13}} & \frac{\gamma_3}{\alpha_{33}}\hat{R}_{23} \\ \frac{\gamma_2}{\alpha_{22}}\hat{R}_{32} & (\gamma_3+\gamma_{23})I_{T_{12}} \ea  \ba{c} \hat{R}_{21}\\ \hat{R}_{31} \ea = 0
\label{fr23}
\eea
Note that the dimension of $ \ba{cc} \hat{R}_{ik}^t & \hat{R}_{jk}^t \ea^t$ is $(T_{jk}+T_{ik})\times T_{ij}$. Therefore nullity of the left matrix in (\ref{refinedRij}) is at least $T_{ij}$. 
Hence if we let the rank of the left matrix in (\ref{refinedRij}) be $r$ we will have:
\begin{equation} \label{r1}
r \leq T_{jk}+T_{ik}-T_{ij}
\end{equation}
On the other hand it is also obvious that:
\begin{eqnarray} \label{r2}
r \geq T_{jk},~ T_{ik}
\end{eqnarray}
From (\ref{r1}) and (\ref{r2}) it follows that:
\begin{equation}
T_{ij} \leq \min({T_{jk},T_{ik}})
\end{equation}
Since a similar argument can be used for $T_{jk}$ and $T_{ik}$ we conclude that:
\begin{equation} \label{eqn:Teq}
T_{12} = T_{23} = T_{13}\triangleq \hat{T}
\end{equation}
Now note that (\ref{fr12})--(\ref{fr23}) is similar to (\ref{RijCond-1})--(\ref{RijCond-3}) with $\hat{R}_{ij}$ instead of $R_{ij}$. Therefore the same argument that led to (\ref{Rsym}) yields that $\hat{R}_{ij}$ is a multiple of an orthogonal matrix say $\Phi_{ij}$; in other words 
\bea
\label{rhat12} \hat{R}_{12} = \hat{R}_{21}^t= \alpha_{12} \Phi_{12}\\
\label{rhat13} \hat{R}_{13} = \hat{R}_{31}^t = \alpha_{13} \Phi_{13}\\
\label{rhat23} \hat{R}_{23} = \hat{R}_{32}^t = \alpha_{23} \Phi_{23}
\eea
where similar to (\ref{alphaij}), $\alpha_{ij}$ is given by
\be
\label{alphaijsign}
\alpha_{ij} = \pm \sqrt{\frac{(\gamma_i+\gamma_{ij})(\gamma_j+\gamma_{ij})}{\gamma_i\gamma_j} \alpha_{ii}\alpha_{jj}} 
\ee
and we have,
\bea
\label{phi-for-rhat}
\Phi_{13} = \begin{cases}
+\Phi_{12}\Phi_{23}~~~~\mbox{if}~~~\alpha_{12}\alpha_{13}\alpha_{23}>0 \\ -\Phi_{12}\Phi_{23}~~~~\mbox{if}~~~\alpha_{12}\alpha_{13}\alpha_{23}<0
\end{cases}
\eea
Note that sign of $\alpha_{ij}$ can be chosen arbitrarily in (\ref{alphaijsign}). 
Finally it follows that after a series of permutations on (\ref{Rstructure}) and substituting for values of $\hat{R}_{ij}$ from (\ref{rhat12})--(\ref{rhat23}), $R' = \Theta^* R \Theta$ can be written as follows:
\begin{equation} \label{boundary-structure}
\ba{cccccc} \alpha_{11} I_{\hat{T}} & \alpha_{12} \Phi_{12} & \alpha_{13} \Phi_{13} &\hspace{-10pt} 0 & \hspace{-10pt}0 & \hspace{-10pt}0\\
\alpha_{12} \Phi_{12}^t & \alpha_{22} I_{\hat{T}} & \alpha_{23} \Phi_{23} &\hspace{-10pt} 0 &\hspace{-10pt} 0 &\hspace{-10pt} 0\\
\alpha_{13} \Phi_{13}^t  & \alpha_{23} \Phi_{23}^t & \alpha_{33} I_{\hat{T}}&\hspace{-10pt} 0 &\hspace{-10pt} 0 &\hspace{-10pt} 0\\
0 & 0 & 0 & \hspace{-10pt} \alpha_{11} I_{T-\hat{T}} & \hspace{-10pt} 0 & \hspace{-10pt} 0 \\
0 & 0 & 0 &\hspace{-10pt} 0 &\hspace{-10pt} \alpha_{22} I_{T-\hat{T}} &\hspace{-10pt} 0 \\
0 & 0 & 0  &\hspace{-10pt} 0 &\hspace{-10pt} 0 &\hspace{-10pt} \alpha_{33} I_{T-\hat{T}}\\
\ea
\end{equation}
which if viewed as the timeshare of a set of Gaussian random variables with an orthogonal covariance matrix and another set of independent random variables, it has the same block principal minors as (\ref{Rstructure}). Moreover $R'$ has the same principal minors as $R$ and $\tilde{R}$ and therefore (\ref{boundary-structure}) is the optimizing solution to problem (\ref{Gaussian_eqn:optimization}). 
However note that (\ref{boundary-structure}) is an optimal solution only if it is a positive semi-definite matrix. Therefore $\alpha_{ij}$'s and $\Phi_{ij}$'s should be such that,\footnote{Note that $\alpha_{ij}$ and the set of $\gamma_i, \gamma_{ij}$ are dependent through (\ref{alphaijsign}). However it can be shown that for a given set of $\alpha_i, \alpha_{ij}$ the values of $\gamma_i, \gamma_{ij}$ can be determined such that (\ref{alphaijsign}) and (\ref{balance1})--(\ref{balance3}) will hold.}
\bea
\alpha_{ii}\alpha_{jj} - \alpha_{ij}^2 \geq 0
\eea
\bea
\det \left(\begin{array}{ccc}\alpha_{11}I_T & \alpha_{12}\Phi_{12} & \alpha_{13}\Phi_{13} \\ \alpha_{12}\Phi_{12}^t & \alpha_{22}I_T & \alpha_{23}\Phi_{23} \\ \alpha_{13}\Phi_{13}^t & \alpha_{23}\Phi_{23}^t & \alpha_{33}I_T \end{array}\right) \geq 0
\eea
\end{proof}

\begin{lemma}[Block Orthogonal, Block Diagonal Covariance] Consider the
  covariance matrix
\be \label{det3-hyper}
R = \ba{ccc} \alpha_{11}I_T & \alpha_{12}\Phi_{12} &
  \alpha_{13}\Phi_{13} \\
\alpha_{12}\Phi_{12}^{t} & \alpha_{22}I_T & \alpha_{23}\Phi_{23} \\
\alpha_{13}\Phi_{13}^{t} & \alpha_{23}\Phi_{23}^{t} & \alpha_{33}I_T\ea ,
\ee
where the $\Phi_{ij}$ are orthogonal, $\alpha_{ii}>0$, and $\alpha_{ii}\alpha_{jj}\geq \alpha_{ij}^2$ in the $2\times 2$ block
principal minors $m_{ij} =
(\alpha_{ii}\alpha_{jj}-\alpha_{ij}^2)^T$. 
Then
\bea
\nonumber  \Big( \alpha_{11} \alpha_{22} \alpha_{33} - \alpha_{11} \alpha_{23}^2
-\alpha_{22} \alpha_{13}^2 -\alpha_{33} \alpha_{12}^2 -2 |\alpha_{12} \alpha_{13} \alpha_{23}| \Big)^{{T}} \leq
 \det R& \\
\leq \Big( \alpha_{11} \alpha_{22} \alpha_{33} - \alpha_{11} \alpha_{23}^2
- \alpha_{22} \alpha_{13}^2 -\alpha_{33} \alpha_{12}^2 + & \hspace{-5pt} 2 |\alpha_{12} \alpha_{13} \alpha_{23}| \Big)^{{T}} \label{eqn:m123}
\eea
where $\Phi = \Phi_{13}^{t} \Phi_{12} \Phi_{23}$ and the upper bound is tight when $\Phi + \Phi^t = 2I$ and the lower bound is achieved when $\Phi + \Phi^t = -2I$. 
\label{lem:bubb}
\end{lemma}

\begin{proof}
We can easily write the following,
\bea
\nonumber \det R & \hspace{-5pt}=& \hspace{-5pt} \frac{1}{\alpha_{11}^T} \det \left( \left(\begin{array}{cc} \alpha_{11}\alpha_{22} I_T & \alpha_{11}\alpha_{23}\Phi_{23} \\ \alpha_{11}\alpha_{23}\Phi_{23}^t & \alpha_{11}\alpha_{33}I_T \end{array}\right) - \left(\begin{array}{c} \alpha_{12}\Phi_{12}^t \\ \alpha_{13}\Phi_{13}^t  \end{array}\right) \left(\begin{array}{cc} \alpha_{12}\Phi_{12} & \alpha_{13}\Phi_{13} \end{array}\right)\right)\\
\nonumber & \hspace{-55pt} =& \hspace{-30pt} \frac{1}{\alpha_{11}^T} \det \left( (\alpha_{11}\alpha_{22}-\alpha_{12}^2)(\alpha_{11}\alpha_{33}-\alpha_{13}^2)I_T-
(\alpha_{11}\alpha_{23}\Phi_{23}^t-\alpha_{12}\alpha_{13}\Phi_{13}^t\Phi_{12})
(\alpha_{11}\alpha_{23}\Phi_{23}-\alpha_{12}\alpha_{13}\Phi_{12}^t\Phi_{13}) \right)\\
& \hspace{-55pt} = & \hspace{-30pt} \det \Big( (\alpha_{11} \alpha_{22} \alpha_{33} - \alpha_{11} \alpha_{23}^2 -\alpha_{22} \alpha_{13}^2 -\alpha_{33} \alpha_{12}^2)I_{T} + \alpha_{12} \alpha_{13} \alpha_{23} (\Phi_{13}^t \Phi_{12}\Phi_{23}+\Phi_{23}^t \Phi_{12}^t \Phi_{13}) \Big)
\end{eqnarray}
The result immediately follows from $-2 I \leq \Phi+\Phi^t \leq 2 I $.
\end{proof}

The $1\times 1$ and $2\times 2$ minors of the covariance matrix structure (\ref{boundary-structure}) obtained in Lemma \ref{lem:boundary} can now be written as,
\bea
&& m_i = \alpha_{ii}^T \label{eqn:mi}\\
&& m_{ij} = (\alpha_{ii}\alpha_{jj}-\alpha_{ij}^2)^{{\hat{T}}}(\alpha_{ii}\alpha_{jj})^{{T-\hat{T}}} \label{eqn:mij}
\eea
Moreover since based on (\ref{phi-for-rhat}), for matrix (\ref{boundary-structure}), $\Phi = \Phi^t_{13}\Phi_{12}\Phi_{23} = \mbox{sign}(\alpha_{12}\alpha_{13}\alpha_{23}) I$, based on Lemma \ref{lem:bubb}, $m_{123}$ of (\ref{boundary-structure}) is given by:
\bea
\label{m123-structure}
(\alpha_{11}\alpha_{22}\alpha_{33})^{T-\hat{T}} \Big( \alpha_{11} \alpha_{22} \alpha_{33} - \alpha_{11} \alpha_{23}^2
-\alpha_{22} \alpha_{13}^2 -\alpha_{33} \alpha_{12}^2 \pm 2 |\alpha_{12} \alpha_{13} \alpha_{23}| \Big)^{\hat{T}} 
\eea
However these values can also be obtained by a timeshare of 3 scalar random variables with covariance matrix,
\bea
\left(\begin{array}{ccc} \alpha_{11} & \alpha_{12} &  \alpha_{13}\\
\alpha_{12} & \alpha_{22} & \alpha_{23}\\ \alpha_{13} & \alpha_{23} & \alpha_{33} \end{array}\right)
\eea
and 3 other independent scalar random variables. This suggests that the region of 3 vector-valued Gaussian random variables may be obtained from the convex hull region of 3 scalar Gaussian random variables. In other words for $n=3$, considering vector-valued random variables will not give any entropy vector that is not obtainable from scalar valued ones. This is essentially the statement of Theorem \ref{thm:vector-vs-scalar} and we can now proceed to a more formal proof.

{\em Proof of Theorem \ref{thm:vector-vs-scalar}:} 
As in Lemma \ref{lem:boundary}, we write the following optimization problem,
\begin{equation}\label{Gaussian_eqn:optimization_makingdiagonal}
\max_{R}\sum_{s\subseteq \{1,2,3\}}{\gamma_s \log\det{\tilde{R}_s}},
\end{equation}
As was obtained in Lemma \ref{lem:boundary}, the optimal solution is of the following form:
\bea
\left(\begin{array}{cccccc} \alpha_{11}I_{\hat{T}} & 0 & \alpha_{12}\Phi_{12} & 0 & \alpha_{13}\Phi_{13} & 0\\
0 & \alpha_{11}I_{T-\hat{T}} & 0 & 0 & 0 & 0\\
\alpha_{12}\Phi_{12}^t & 0 & \alpha_{22}I_{\hat{T}} & 0 & \alpha_{23}\Phi_{23} & 0\\
0 &0 & 0 & \alpha_{22}I_{T-\hat{T}} & 0 & 0\\
\alpha_{13}\Phi_{13}^t & 0 & \alpha_{23}\Phi_{23}^t & 0 & \alpha_{33}I_{\hat{T}} & 0\\
0 & 0 & 0 & 0 & 0 & \alpha_{33}I_{\hat{T-\hat{T}}} 
\end{array}\right)
\eea
where $\Phi_{ij}$ are orthogonal matrices and $\Phi_{13} = \mbox{sign}(\alpha_{12}\alpha_{13}\alpha_{23}) \Phi_{12}\Phi_{23}$. 
Now let
\bea
\mathbb{Q} = \left(\begin{array}{cccccc} I_{\hat{T}} & 0 & 0 & 0 & 0 & 0 \\
0 & I_{T-\hat{T}} & 0 & 0 & 0 & 0\\
0 & 0 & {\Phi}_{12}^t & 0 & 0 & 0\\
0 & 0 & 0 & I_{T-\hat{T}} & 0 & 0\\
0 & 0 & 0 & 0 & {\Phi}_{13}^t & 0\\
0 & 0 & 0 & 0 & 0 & I_{T-\hat{T}}
\end{array}\right)
\eea
and define
\bea
{R}^{\mathbb{Q}} \triangleq \mathbb{Q}^t R \mathbb{Q} = 
\left(\begin{array}{cccccc} \alpha_{11}I_{\hat{T}} & 0 & \alpha_{12}I_{\hat{T}} & 0 & \alpha_{13}I_{\hat{T}} & 0\\
0 & \alpha_{11}I_{T-\hat{T}} & 0 & 0 & 0 & 0\\
\alpha_{12}I_{\hat{T}} & 0 & \alpha_{22}I_{\hat{T}} & 0 & \alpha_{23}\Phi_{12}\Phi_{23}\Phi_{13}^t & 0\\
0 &0 & 0 & \alpha_{22}I_{T-\hat{T}} & 0 & 0\\
\alpha_{13}I_{\hat{T}} & 0 & \alpha_{23}\Phi_{13}\Phi_{23}^t \Phi_{12}^t & 0 & \alpha_{33}I_{\hat{T}} & 0\\
0 & 0 & 0 & 0 & 0 & \alpha_{33}I_{\hat{T-\hat{T}}} 
\end{array}\right)
\eea
However since $\Phi_{12}\Phi_{23}\Phi_{13}^t = \pm I_{\hat{T}}$, it follows that all blocks of $R^{\mathbb{Q}}$ are diagonal and therefore $R^{\mathbb{Q}}$ can be viwed as a timeshare of scalar random variables.  
Moreover since $R^{\mathbb{Q}}$ has the same principal minors as $R$, therefore $R^{\mathbb{Q}}$ is also an optimal solution of (\ref{Gaussian_eqn:optimization_makingdiagonal}). 
\hfill \qed

In order to proceed to the proof of Theorem \ref{thm:cone} we need the following lemma,




\begin{lemma} Consider the function
\be \label{f-definition}
f(\delta) = \left(\max\left[0, -2+\sum_{l=1}^3x_l^{\delta}+
2\sqrt{\prod_{l=1}^3(1-x_l^{{\delta}})}\right]\right)^{\frac{1}{\delta}},
\ee
where $0<x_l\leq 1$, for $l=1,2,3$. $f$ is either a constant function equal to $\min_{i\neq j} x_ix_j$ or has a unique global maximum given by:
\be
\max_{\delta} f(\delta) = \frac{\prod_{l}x_l}{\max_{l}(x_l)} = \min_{i,j \in \{1,2,3\}, i\neq j} x_ix_j
\ee

Moreover if we let \mbox{$\tilde{y} = \frac{\prod_l x_l}{\max_l x_l}+2\max_l x_l-\sum_l x_l$}, then
\bea
\mbox{If $\tilde{y}\leq 0$}~~~\Longrightarrow~~~\max_{\delta\geq 1} f(\delta) = \frac{\prod_{l}x_l}{\max_{l}(x_l)}  
\eea
\label{lem:f}
\end{lemma}

\begin{proof}
See Appendix.
\end{proof}

\begin{corollary}
\label{cor:ftheta}
Let $\theta = \frac{1}{\delta}$, in Lemma \ref{lem:f}. Then $f(\frac{1}{\theta})$ is either a constant function equal to $\min_{i\neq j} x_ix_j$ or has a unique global maximizer such that $\max_{\theta} f(\frac{1}{\theta}) = \min_{i\neq j} x_ix_j$. Furthermore,
\be
\mbox{If $\tilde{y}\leq 0$}~~~\Longrightarrow~~~\max_{0\leq \theta\leq 1} f(\frac{1}{\theta}) = \frac{\prod_{l}x_l}{\max_{l}(x_l)}  
\ee
\end{corollary}

%

%

Now we can proceed to the proof of Theorem \ref{thm:cone}.

{\em Proof of Theorem \ref{thm:cone}:} We show that the entropy region of 3 continuous random variables can be generated from the convex cone of the Gaussian entropy region. To this end we prove that any entropy vector of 3 continuous random variables lies in the convex cone of Gaussian entropies. Let $g$ be an arbitrary entropy vector corresponding to 3 continuous random variables. We know that the only inequalities that constrain the entries of $g$ are
\be
\label{eqn:gchar}
g_{ij}\leq g_i+g_j,~~~g_{123}+g_k\leq g_{ik}+g_{jk}
\ee
Let $p = e^g$ where the exponential acts componentwise. Then the equivalent set of constraints to (\ref{eqn:gchar}) are
\be
\label{eqn:cond-on-p}
p_i\geq 0,~~~0\leq p_{ij}\leq p_ip_j,~~~0\leq p_{123} \leq \frac{p_{ik}p_{jk}}{p_k}
\ee
We now show that any such $p$-vector can be obtained from Gaussian random variables. Consider the structure (\ref{boundary-structure}) obtained in Lemma \ref{lem:boundary} which suggested the time-share of a set of independent random variables with covariance matrix of block size $T-\hat{T}$ and another set of random variables with orthogonal covariance matrix of block size $\hat{T}$. We try to find $\alpha_{ii},~\alpha_{ij}$ in structure (\ref{boundary-structure}) that will yield the desired $p_i$ and $p_{ij}$. Therefore if $m$ is the vector of block principal minors of structure (\ref{boundary-structure}) we need to obtain $m^{\frac{1}{T}} = p$. Using (\ref{eqn:mi}) and (\ref{eqn:mij}) we can solve for $\alpha_{ii}$ and $\alpha_{ij}$ and obtain
\be
\label{eqn:alphap}
\alpha_{ii} = p_i >0,~~~\alpha_{ij} = \pm \sqrt{ p_ip_j (1 - (\frac{p_{ij}}{p_ip_j})^{\frac{T}{\hat{T}}})}
\ee

Now we need to show that $p_{123}$ falls within the set of achievable values of  
$m_{123}^{\frac{1}{T}}$. 
Note that calculating the determinant of the matrix in (\ref{boundary-structure}) via Lemma \ref{lem:bubb} or equation (\ref{m123-structure}) gives
\begin{eqnarray} \label{maxdet}
\max \mbox{det}R =  (\alpha_{11}\alpha_{22}\alpha_{33})^{T-\hat{T}} \Big( \alpha_{11} \alpha_{22} \alpha_{33} - \alpha_{11} \alpha_{23}^2
-\alpha_{22} \alpha_{13}^2 -\alpha_{33} \alpha_{12}^2 +2 |\alpha_{12} \alpha_{13} \alpha_{23}| \Big)^{\hat{T}}
\end{eqnarray}
where the $\max$ is achieved when $\alpha_{12}\alpha_{13}\alpha_{23}\geq 0$.
Letting $\theta = \frac{\hat{T}}{T}$ and replacing for values of $\alpha_{ii}$ and $\alpha_{ij}$ in terms of $p_i$ and $p_{ij}$ from (\ref{eqn:alphap}) yields
\begin{eqnarray} \label{m123}
\nonumber  \max m_{123}^{\frac{1}{T}} = p_1 p_2 p_3 \Bigg( && \hspace{-20pt} -2+\hspace{-3pt} \left(\frac{p_{12}}{p_1 p_2}\right)^{\frac{1}{\theta}}\hspace{-5pt} + \hspace{-3pt} \left(\frac{p_{13}}{p_1 p_3}\right)^{\frac{1}{\theta}}\hspace{-5pt} + \hspace{-3pt} \left(\frac{p_{23}}{p_2 p_3}\right)^{\frac{1}{\theta}} \\
 && \hspace{-20pt}  +2 \sqrt{ \left( 1 - \hspace{-3pt}\left(\frac{p_{12}}{p_1 p_2}\right)^{\frac{1}{\theta}} \right) \left( 1 -\hspace{-3pt} \left(\frac{p_{13}}{p_1 p_3}\right)^{\frac{1}{\theta}} \right) \left( 1 - \hspace{-3pt}\left(\frac{p_{23}}{p_2 p_3}\right)^{\frac{1}{\theta}} \right) }\Bigg)^{\theta}
\end{eqnarray}
Of course this corresponds to the determinant of a covariance matrix of the time-share of some Gaussian random variables only if the term inside the outer parenthesis in (\ref{m123}) is positive. Therefore assuming $x_1= \frac{p_{23}}{p_2p_3},~
x_2= \frac{p_{13}}{p_1p_3},~
x_3= \frac{p_{12}}{p_1p_2}
$
and using (\ref{f-definition}) in Lemma \ref{lem:f}:
\be \label{supf}
\max m_{123}^{\frac{1}{T}} = p_1p_2p_3 f(\frac{1}{\theta})
\ee
It remains to show that for any given $p_{123}$ that satisfies the latter condition of (\ref{eqn:cond-on-p}), we have 
$p_{123}\leq p_1p_2p_3 \sup_{\Phi,\theta} f(\frac{1}{\theta})$. Since by (\ref{eqn:cond-on-p}), $p_{123}$ can be as large as $\min \left(\frac{p_{ik}p_{jk}}{p_k}\right)$ we really need to show that $p_1p_2p_3 \sup f(\frac{1}{\theta})$ achieves $\min \left(\frac{p_{ik}p_{jk}}{p_k}\right)$ for some value of $\theta$.
Therefore we need to compute,
\be
\sup_{\theta} m_{123}^{\frac{1}{T}} = p_1p_2p_3 \sup_{0<\theta\leq 1} f(\frac{1}{\theta})
\ee
Note that since we have fixed $p_i$ and $p_{ij}$, and that $\theta$ represents the timesharing of 2 sets of random variables, $\theta=0$ is not generally allowed (otherwise we enforce the random variables to be independent which is not necessarily the case for given $p_i$ and $p_{ij}$). Therefore we have used $\sup$ instead of $\max$ in (\ref{supf}). 
To find $\sup f(\frac{1}{\theta})$ with respect to $\theta$ over \mbox{$0< \theta \leq 1$},
note that as stated in Lemma \ref{lem:f}, $f(\frac{1}{\theta})$ has a (unique) global maximum with a value of $\min_{i\neq j} x_i x_j = \min \frac{p_{jk}p_{ik}}{p_ip_jp_k^2}$ at some $\theta = \theta_0$. If for the assumed values of $x_1,x_2,x_3$ (obtained from the fixed values of $p_i$ and $p_{ij}$), $f(\frac{1}{\theta})$ achives its maximum for $0\leq \theta\leq 1$, i.e., $0\leq \theta_0\leq 1$ then 
\be
\sup_{\Phi,~0<\theta\leq 1} m_{123}^{\frac{1}{T}} = p_1p_2p_3 \min \frac{p_{jk}p_{ik}}{p_ip_jp_k^2} = \min \frac{p_{jk}p_{ik}}{p_k}
\ee
This immediately gives that the vector $p = e^g$ is achievable. 
Otherwise if $\theta_0>1$, then for some $\theta'>\theta_0$, define the vector $p' = p^{\frac{1}{\theta'}}$ (elementwise exponentiation). This means that $p'_i= p_i^{\frac{1}{\theta'}} , p'_{ij}= p_{ij}^{\frac{1}{\theta'}}$ and hence $x'_i = x_i^{\frac{1}{\theta'}}$. Now we try to achieve vector $p'$ by Gaussian structure of (\ref{boundary-structure}). For this purpose we follow similar steps as above for $p'$ and let $m'$ to be the vector of block principal minors of the new corresponding matrix. Then we have
\be
\label{eqn:mprimelimit}
\sup_{\Phi,~0<\theta\leq 1} {m'}_{123}^{\frac{1}{T}} = p'_1p'_2p'_3 \max_{0\leq \theta\leq 1} \left( f(\frac{1}{\theta\theta'}) \right)^{\frac{1}{\theta'}}
\ee
The global maximum of $f$ will now happen for $\theta = \frac{\theta_0}{\theta'}<1$ at which it will have the value $(\min \frac{p_{jk}p_{ik} }{p_k})$. Replacing this in (\ref{eqn:mprimelimit}) gives $\sup_{0<\theta\leq 1} {m'}_{123}^{\frac{1}{T}} = \min \frac{p'_{jk}p'_{ik} }{p'_k}$ which means that although $p$ was not achievable with Gaussians $p' = p^{\frac{1}{\theta}}$ is achievable. The result of the theorem is then established by noting that $p'$ corresponds to a valid entropy vector $g' = \frac{1}{\theta'} g$, i.e., a scaled version of $g$. Note that if maximum of $f$ happens at infinity, i.e., $\theta
_0\rightarrow \infty$, then we should consider a sequence of scaled vectors $p'_i = p^{\frac{1}{\theta'_i}}$ (or equivalently $g'_i = \frac{1}{\theta'_i}g$) where $\theta'_i$ is an unbounded increasing sequence in $i$. As $i \rightarrow \infty$, $g'_i$ will asymptotically fall in the Gaussian region (a small perturbation of $g'_i,~ i\rightarrow \infty$ will put $g'_i$ in the Gaussian region).  Hence $g$ will belong to the closure of the convex cone of the Gaussian region as well. 
\hfill \qed
 



\begin{conjecture}
In Lemma \ref{lem:f}, if for some $\tilde{\delta}>0$, we have $y(\tilde{\delta})>0$, then $\max_{\delta\geq \tilde{\delta}} f(\delta) = f(\tilde{\delta})$. In particular
\be
\label{eqn:fconj}
\mbox{If $\tilde{y} = y(1) > 0$}~~~\Longrightarrow~~~\max_{\delta\geq1} f(\delta) = f(1)  
\ee
Moreover if we let $\delta = \frac{1}{\theta}$, then (\ref{eqn:fconj}) translates to
\be
\mbox{If $\tilde{y} = y(1) > 0$}~~~\Longrightarrow~~~\max_{0\leq \theta\leq1} f(\frac{1}{\theta}) = f(1)  
\ee
\label{conj:func-conj}
\end{conjecture}
Simulations of the function $f$ for different values of $\delta$ (Fig. \ref{Gaussian_fig:fydelta} in the Appendix) support the statement of Conjecture \ref{conj:func-conj}. Our Conjecture \ref{thm:3Gauss} relies on the above.

{\em Proof of Conjecture \ref {thm:3Gauss} Assuming Conjecture \ref{conj:func-conj}:}
To find the Gaussian entropy region again we employ Lemma \ref{lem:boundary} to obtain the boundary entropies of the region. 
Hence we consider the structure of (\ref{boundary-structure}) which is obtained from the time-share of a set of independent random variables with covariance matrix of block size $T-\hat{T}$ and another set of random variables with orthogonal covariance matrix of block size $\hat{T}$.
Let $m$ be the vector of block principal minors of the matrix in (\ref{boundary-structure}) and let $q = m^{\frac{1}{T}}$ where the exponential acts componentwise.
Moreover denote the corresponding entropy vector by $g = \log q$ ($\log$ acting componentwise).
Then we would like to characterize the set of $q$-vectors (equivalently $g$-vectors) that can arise from (\ref{boundary-structure}) (i.e., they lie in the convex hull of Gaussian structures). 
First let us investigate the constraints on $q_i$ and $q_{ij}$.
It is easy to see that $
\alpha_{ii} = q_i \geq 0$ and \mbox{$ \alpha_{ij} = \pm \sqrt{ q_iq_j (1 - (\frac{q_{ij}}{q_iq_j})^{\frac{T}{\hat{T}}})}$}. Therefore the imposed constraints are
\be \label{pipj}
q_i\geq 0,~~~~0\leq q_{ij}\leq q_iq_j
\ee
where $q_{ij}\geq 0$ is due to the positivity requirement of the matrix. Next we would like to obtain the limits of $q_{123}$. For this purpose assume that $q_i$ and $q_{ij}$ are given fixed numbers satisfying (\ref{pipj}).
Determinant of matrix (\ref{boundary-structure}) is obtained from (\ref{m123-structure}). 
If we let $\theta=\frac{\hat{T}}{T}$ and substitute for $\alpha_i$ and $\alpha_{ij}$ in (\ref{m123-structure}) we obtain:
\begin{eqnarray} \label{q123}
\nonumber  \max q_{123} = q_1 q_2 q_3 \Bigg( && \hspace{-20pt} -2+\hspace{-3pt} \left(\frac{q_{12}}{q_1 q_2}\right)^{\frac{1}{\theta}}\hspace{-5pt} + \hspace{-3pt} \left(\frac{q_{13}}{q_1 q_3}\right)^{\frac{1}{\theta}}\hspace{-5pt} + \hspace{-3pt} \left(\frac{q_{23}}{q_2 q_3}\right)^{\frac{1}{\theta}} \\
 && \hspace{-20pt}  +2 \sqrt{ \left( 1 - \hspace{-3pt}\left(\frac{q_{12}}{q_1 q_2}\right)^{\frac{1}{\theta}} \right) \left( 1 -\hspace{-3pt} \left(\frac{q_{13}}{q_1 q_3}\right)^{\frac{1}{\theta}} \right) \left( 1 - \hspace{-3pt}\left(\frac{q_{23}}{q_2 q_3}\right)^{\frac{1}{\theta}} \right) }\Bigg)^{\theta}
\end{eqnarray}
Note that as described in the proof of Theorem {\ref{thm:cone}}, we should insist on the positivity of the expression inside the outermost parenthesis in (\ref{q123}). Therefore defining 
$x_1= \frac{q_{23}}{q_2q_3},~
x_2= \frac{q_{13}}{q_1q_3},~
x_3= \frac{q_{12}}{q_1q_2}
$
and using (\ref{f-definition}) in Lemma \ref{lem:f}:
\be \label{supq}
\sup_{0<\theta\leq 1} q_{123} = q_1q_2q_3 \max_{0\leq \theta\leq 1} f(\frac{1}{\theta})
\ee
where since $q_i$ and $q_{ij}$ are fixed, we have excluded $\theta = 0$ and used $\sup$ instead of $\max$ for $q_{123}$ in (\ref{supq}) so that we do not enforce indepndence of the underlying random variables.
\footnote{Note that if the underlying random variables are independent, then $f(\theta)$ will be independent of $\theta$.} 
Using Corollary \ref{cor:ftheta} and assuming that Conjecture \ref{conj:func-conj} holds we obtain,
\bea
\max_{0\leq\theta\leq 1} f(\frac{1}{\theta}) = \begin{cases}  \frac{\prod_{l}x_l}{\max_{l}(x_l)} ~~~ \mbox{If $\tilde{y}\leq 0$}\\ f(1) \hspace{25pt} \mbox{If $\tilde{y}>0$}\end{cases}
\label{eqn:fmax}
\eea
where $\tilde{y} =  \frac{\prod_l x_l}{\max_l x_l}+2\max_l x_l-\sum_l x_l $. Replacing for $x_i$ in (\ref{eqn:fmax}) in terms of $q$-vector entries and using the result of (\ref{eqn:fmax}) in (\ref{supq}) with the final substitution of $q$-entries in terms of entropy elements of $g = \log q$, yields the cojecture result. Note that when $\tilde{y}<0$ the characterization of the region is known perfectly from Corollary \ref{cor:ftheta}. The only part of the entropy region that is conjectured about is when $\tilde{y}>0$ whose proof relies on the validity of Conjecture \ref{conj:func-conj}.  
\hfill \qed



\section{Cayley's Hyperdeterminant} \label{Gaussian_sec:hyperdet}
Recall from (\ref{defg}) that the entropy of a collection of Gaussian random variables
is simply the ``log-determinant'' of their covariance
matrix. Similarly, the entropy of any {\it subset} of Gaussian random variables
is simply the ``log'' of the principal minor of the covariance matrix corresponding to this
subset. Therefore one approach to characterizing the
entropy region of Gaussians, is to study the determinantal relations of
a symmetric positive semi-definite matrix.

For example, consider 3 Gaussian random
variables. While the entropy vector of 3 random variables is a 7 dimensional
object, there are only 6 free parameters in a symmetric positive
semi-definite matrix and therefore the minors should satisfy a constraint.
It has very recently been shown that this constraint is given by Cayley's
so-called $2\times 2\times 2$ ``hyperdeterminant" \cite{hyperdet}.
The hyperdeterminant is a generalization of the
determinant concept for matrices to tensors and it was first introduced by
Cayley in 1845 \cite{cayley}.

There are a couple of equivalent definitions for the hyperdeterminant
among which we choose the definition through the degeneracy of a multilinear
form. Consider the following multilinear form of the format $(k_1+1)
\times (k_2+1) \times \ldots \times (k_n+1)$ in variables
$X_1,\ldots,X_n$ where each variable $X_j$ is a vector of length
$(k_j+1)$ with elements in $\mathds{C}$:
\begin{eqnarray}\label{multilinear}
\nonumber  &&\hspace{-40pt} f(X_{1},X_{2},\ldots,X_{n}) = \\
&& \hspace{-10pt} \sum_{i_1=0}^{k_1}\sum_{i_2=0}^{k_2}\ldots\sum_{i_n=0}^{k_n} a_{i_1,i_2,\ldots,i_n} x_{1,i_1} x_{2,i_2},\ldots,x_{n,i_n}
\end{eqnarray}
The multilinear form $f$ is said to be degenerate if and only if there
is a non-trivial solution $(X_1,X_2,\ldots,X_n)$ to the following
system of partial derivative equations \cite{gelfand}:
\begin{equation}\label{degeneracy}
\frac{\partial f}{\partial x_{j,\hspace{1pt}i}} = 0 ~~~ \text{for all} ~~ j=1,\ldots,n ~~\text{and}~~ i=1,\ldots,k_j
\end{equation}
The unique (up to a scale) irreducible polynomial function of entries $a_{i_1,i_2,\ldots,i_n}$ with integral
coefficients that vanishes when $f$ is degenerate is called the hyperdeterminant.

\textit{Example ($2\times 2$ hyperdeterminant):}
Let $X = \left( \begin{array}{c} x_0 \\ x_1 \end{array}\right), Y = \left( \begin{array}{c} y_0 \\ y_1 \end{array}\right),$ and 
$A = \left( \begin{array}{cc} a_{00} & a_{01}\\ a_{10} & a_{11} \end{array}\right)$. 
Consider the multilinear form $f(X,Y) = \sum_{i,j = 0}^1 a_{i,j}x_iy_j = X^t A Y$.
The multilinear form $f$ is degenerate if there is a non-tirivial solution for $X,Y$ such that
\bea
&& \frac{\partial f}{\partial X} = AY = 0\\
&& \frac{\partial f}{\partial Y} = A^tX = 0
\eea
A nontirival solution exists if and only if $\det A = 0$. Therefore the hyperdeterminant is simply the determinant in this case.


The hyperdeterminant of a $2\times 2 \times 2$ multilinear form $\sum_{i_1 = 0}^1 \sum_{i_2=0}^1 \sum_{i_3=0}^1 a_{i_1i_2i_3} x_{i_1}x_{i_2}x_{i_3}$ was
first computed by Cayley \cite{cayley} and is as follows:
\begin{eqnarray}\label{hyperdet_minor}
\nonumber &&\hspace{-40pt} -a_{000}^2 a_{111}^2 - a_{100}^2a_{011}^2-a_{010}^2a_{101}^2-a_{001}^2a_{110}^2  \\
\nonumber && \hspace{-30pt} -4a_{000} a_{110} a_{101}a_{011} - 4a_{100}a_{010}a_{001}a_{111}\\
\nonumber && \hspace{-20pt} +2 a_{000} a_{100}a_{011}a_{111} + 2 a_{000} a_{010} a_{101} a_{111}\\
\nonumber && \hspace{-10pt} + 2a_{000} a_{001} a_{110} a_{111}+2 a_{100}a_{010}a_{101}a_{011}\\
&&\hspace{0pt} +2a_{100}a_{001}a_{110}a_{011}+2a_{010}a_{001}a_{110}a_{101} = 0
\end{eqnarray}

In \cite{hyperdet} it is further shown that the principal
minors of an $n\times n$ symmetric matrix satisfy the
$2\times 2\times \ldots \times 2$ ($n$ times)
hyperdeterminant. It is thus clear that determining the entropy region
of Gaussian random variables is intimately related to Cayley's
hyperdeterminant.

It is with this viewpoint in mind that we study the
hyperdeterminant in this section.  
In the next 2 subsections, first we present a 
a new determinant formula for the $2\times 2\times 2$
hyperdeterminant which may be of
interest since computing the hyperdeterminant of higher formats is
extremely difficult and our formula may suggest a way of attacking
more complicated hyperdeterminants. Then we give a novel proof of one of the
main results of \cite{hyperdet} that the principal
minors of any $n\times n$ symmetric matrix satisfy the
$2\times 2\times \ldots \times 2$ ($n$ times)
hyperdeterminant. Our proof hinges on identifying a determinant
formula for the multilinear form from which the hyperdeterminant
arises. 

\subsection{A Formula for the $2\times2 \times 2$ Hyperdeterminant}
Obtaining an explicit formula for the hyperdeterminant is not an easy
task. The first nontrivial hyperdeterminant which is the $2\times 2\times 2$, was obtained by Cayley in 1845 \cite{cayley}. However surprisingly calculating the next hyperdeterminant which is the $2\times 2\times 2\times 2$ proves to be very difficult. Until recently the only method for computing the $2\times 2\times 2\times 2$ was the  nested formula of Schl\"{a}fli which he had obtained in 1852 \cite{schalfli}\cite{gelfand}
and although after 150 years Luque and Thibon \cite{LT-hyperdet-invariants} expressed it in terms of the fundamental tensor invariants, the monomial expansion of this hyperdeterminant remained as a challenge. It was finally solved recently in \cite{hyper4cube} where they show that the $2\times 2\times 2\times 2$ hyperdeterminant consists of 2,894,276 terms. It is interesting to mention that Cayley had a 340 term expression for the $2\times 2\times 2\times 2$ hyperdeterminant which satisfies many invariance properties of the hyperdeterminant and only fails to satisfy a few extra conditions \cite{TW-hyper-integrable}. Therefore, as mentioned previously, computing hyperdeterminants of different formats is generally nontrivial. In fact even Schl\"{a}fli's method only works for some special hyperdeterminant formats.
Moreover according to \cite{gelfand} it is not easy to prove directly that (\ref{hyperdet_minor}) vanishes if and only if (\ref{degeneracy}) has a non-trivial solution.
Here we propose a new formula for (and a method to obtain) the $2\times
2\times 2$ hyperdeterminant which shows this is an if and only if connection directly. Moreover this method might be extendable to hyperdeterminants of larger format.

\begin{theorem} \textit{(Determinant formula for $2\times 2\times 2$ hyperdeterminant)} Define
\[ B_0 = \left[\hspace{-2pt} \begin{array}{cc} a_{000} & \hspace{-5pt} a_{100} \\
a_{001} & \hspace{-5pt} a_{101} \end{array} \hspace{-2pt}\right] ~,~
B_1 = \left[\hspace{-2pt} \begin{array}{cc} a_{010} & \hspace{-5pt} a_{110} \\
a_{011} & \hspace{-5pt} a_{111} \end{array} \hspace{-2pt}\right] ~,~
J = \left[\hspace{-2pt} \begin{array}{cc} 0 & \hspace{-5pt} -1 \\ 1 & \hspace{-5pt} 0 \end{array} \hspace{-2pt}\right] \]
Then the $2\times 2\times 2$ hyperdeterminant is given by
\be \label{new_hyper1}
\mbox{det}(B_0JB_1^t-B_1JB_0^t) = 0
\ee \label{hyperdet-formula-det}
\end{theorem}
\begin{proof}
Let $f$ be a multilinear form of the format $2\times 2 \times 2$,
\begin{equation}
f(X,Y,Z) = \sum_{i,j,k=0}^{1} a_{ijk} x_i y_j z_k
\end{equation}
Then by the change of variables, $w_0 = x_0 y_0~,~ w_1 = x_1y_0~,~ w_2 = x_0y_1~,~ w_3 = x_1y_1$,
the function $f$ can be written as,
\begin{eqnarray} \label{matrix_form}
\nonumber  f(X,Y,Z) =&&\hspace{-15pt}(\hspace{-2pt} \begin{array}{cc} z_0 &\hspace{-5pt} z_1 \end{array}\hspace{-2pt}) \left( \hspace{-2pt}\begin{array}{cccc} a_{000} &\hspace{-5pt} a_{100} &\hspace{-5pt} a_{010} &\hspace{-5pt} a_{110} \\ a_{001} &\hspace{-5pt} a_{101} &\hspace{-5pt} a_{011} &\hspace{-5pt} a_{111} \end{array}\hspace{-2pt} \right) \left( \hspace{-2pt}\begin{array}{c} w_0 \\ w_1 \\ w_2 \\ w_3 \end{array}\hspace{-2pt} \right) \\
 \triangleq&& \hspace{-15pt} Z^t \left(\begin{array}{cc}B_0 & B_1 \end{array}\right)  W
\end{eqnarray}

To proceed, recall from (\ref{degeneracy}) that the hyperdeterminant of the multilinear form of the format $2\times 2 \times 2$,
vanishes if and only if there is a non-trivial solution $(X,Y,Z)$ to the system of partial derivative equations:
\begin{equation} \label{degenracy_conditions}
\frac{\partial f}{\partial x_i} = 0 ~~~~~\frac{\partial{f}}{\partial{y_j}} = 0 ~~~~~\frac{\partial{f}}{\partial{z_k}} = 0 ~~~~~i,j,k = 0,1
\end{equation}
(a) First we show that if there is a non-trivial solution to the equations (\ref{degenracy_conditions}), then (\ref{new_hyper1}) vanishes. By the chain rule $\frac{\partial f}{\partial x_i} = \sum_{k}{\frac{\partial w_k}{\partial x_i} \frac{\partial f}{\partial w_k}}$, we can write $\frac{\partial f}{\partial{(X,Y)}} = \left(\frac{\partial W}{\partial {(X,Y)}}\right)^{t} \frac{\partial f}{\partial W}$. Also from (\ref{matrix_form}), $\frac{\partial f}{\partial Z} =  (\begin{array}{cc} B_0 & B_1 \end{array}) W$.
Therefore the degeneracy conditions equivalent with (\ref{degenracy_conditions}) become:
\begin{eqnarray}
\label{null_cond}\left(\frac{\partial W}{\partial {(X,Y)}}\right)^{t} \frac{\partial f}{\partial W} &=& 0\\
\label{eqn2}  (\begin{array}{cc} B_0 & B_1 \end{array}) W &=& 0
\end{eqnarray}
Condition (\ref{null_cond}) implies that the vector $\frac{\partial f}{\partial W}$ should belong to the null space of the matrix $\left(\frac{\partial W}{\partial {(X,Y)}}\right)^{t} $.
The following Lemma gives the structure of this null space.
\begin{lemma}\label{lem:hyper_formula_null}
The null space of the matrix $\left(\frac{\partial W}{\partial {(X,Y)}}\right)^{T} $ is characterized by vectors of the form,
$(w_3~-w_2~-w_1~w_0)^t$.
\end{lemma}

\begin{proof}
Let $V$ be a $4\times 1$ vector. Noting that for $j=\{1,2\}$, $\left(\frac{\partial W}{\partial (X,Y)}\right)_{ij} = \frac{\partial w_i}{\partial x_{j-1}}$ and for $j=\{3,4\}$, $\left(\frac{\partial W}{\partial (X,Y)}\right)_{ij} = \frac{\partial w_i}{\partial y_{j-3}}$, we have:
\begin{eqnarray} \nonumber
& \left(\frac{\partial W}{\partial {(X,Y)}}\right)^{t} V =
\left(\begin{array}{cccc} y_0 & 0 & y_1 & 0 \\ 0 & y_0 & 0 & y_1 \\ x_0 & x_1 & 0 & 0 \\ 0 & 0 & x_0 & x_1 \end{array}\right) \left(\begin{array}{c} v_1 \\ v_2 \\ v_3 \\ v_4 \end{array}\right) = 0
\\ &
\end{eqnarray}
Solving for $V$ in the above, yields the equations:
\bea
\frac{v_1}{v_3} = \frac{v_2}{v_4} = -\frac{y_1}{y_0}\\
\frac{v_1}{v_2} = \frac{v_3}{v_4} = -\frac{x_1}{x_0}
\eea
Letting $v_4 = x_0 y_0$ characterizes the vectors in the null space up to a scale:
\begin{eqnarray} \label{null_space}
\nonumber \hspace{-20pt} V^t &=& (\begin{array}{cccc} x_1y_1 & -x_0y_1 & -x_1y_0 & x_0y_0  \end{array})^t \\
\hspace{-20pt} &=& \left( \begin{array}{cccc}w_3 & -w_2 & -w_1 & w_0  \end{array}\right)^t
\end{eqnarray}
We further note that provided $\left(\begin{array}{c} x_0 \\ x_1 \end{array}\right)\neq 0$ and $\left(\begin{array}{c} y_0 \\ y_1 \end{array}\right)\neq 0$, the matrix $\left(\frac{\partial W}{\partial {(X,Y)}}\right)^{t}$ has rank 3 and that $V$ is therefore the only null space vector (up to a scaling).
\end{proof}

Going back to the proof of Theorem \ref{hyperdet-formula-det}, using Lemma \ref{lem:hyper_formula_null} we conclude that we should have, $\frac{\partial f}{\partial Z} =   (\begin{array}{cc} B_0 & B_1 \end{array}) W = 0$
and for an arbitrary non-zero scalar $\alpha$, $\frac{\partial f}{\partial W } =    (\begin{array}{cc} B_0 & B_1 \end{array})^{t}Z = \alpha \left( \begin{array}{cccc}w_3 & -w_2 & -w_1 & w_0  \end{array}\right)^t$. Putting these two equations into matrix form we can further write the following:
\begin{equation}
\left( \begin{array}{ccc} 0 & 0 & {B_0}^{t}\\0 & 0 & {B_1}^t \\ B_0 & B_1 & 0 \end{array}\right) \left( \begin{array}{c} W \\ Z \end{array}\right) = \alpha \left(\begin{array}{c} w_3 \\-w_2 \\ -w_1\\ w_0\\ 0 \\0  \end{array} \right)
\end{equation}
or in other form:
\begin{equation}\label{W_null}
\left( \begin{array}{ccc} \left(\begin{array}{cc} 0 & 0 \\ 0 &  0  \end{array}\right) & \alpha \left(\begin{array}{cc} 0 & -1 \\ 1 & 0  \end{array}\right) & {B_0}^t\\ \alpha \left(\begin{array}{cc} 0 & 1 \\ -1 & 0  \end{array}\right) & \left(\begin{array}{cc} 0 & 0 \\ 0 & 0  \end{array}\right) & {B_1}^{t} \\ B_0 & B_1 & 0 \end{array} \right) \left(\begin{array}{c} W \\ Z \end{array} \right) = 0
\end{equation}
A non-trivial solution for $X,Y,Z$ and hence for $W,Z$ requires the matrix to be low rank. Evaluating the determinant we have:
\bea
\hspace{-15pt}\det \left(\begin{array}{ccc} 0 & \alpha J & B_0^t\\ -\alpha J & 0 & B_1^t \\ B_0 & B_1 & 0\end{array}\right) & \hspace{-8pt}=& \hspace{-8pt} \det \left(\begin{array}{cc} 0 & \alpha J \\ -\alpha J & 0\end{array}\right) \det \left(- \left(\begin{array}{cc}B_0 & B_1\end{array}\right) \left(\begin{array}{cc} 0 & \alpha J\\ -\alpha J & 0 \end{array} \right)^{-1} \left(\begin{array}{c} B_0^t\\ B_1^t\end{array} \right)\right)\\
&\hspace{-8pt}=&\hspace{-8pt}  \alpha^2 \det  \left( \left(\begin{array}{cc}B_0 & B_1\end{array}\right) \left(\begin{array}{cc} 0 & J\\ - J & 0 \end{array} \right)^{-1} \left(\begin{array}{c} B_0^t\\ B_1^t\end{array} \right)\right) = 0
\eea

Using the fact that $J^{-1} = -J$ we can write the following,
\begin{eqnarray}\label{new_hyper2}
\det\left(  (\begin{array}{cc} B_0 & B_1 \end{array}) \left(\begin{array}{cc} 0 &  J \\ -J & 0 \end{array}\right) \left(\begin{array}{c} {B_0}^t \\ {B_1}^t\end{array}\right) \right)
= \det (B_0 J {B_1}^t - B_1 J {B_0}^t ) = 0
\end{eqnarray}

Note that the explicit calculation of (\ref{new_hyper2}) gives,
\bea
\det \left(\begin{array}{cc} 2(a_{100}a_{010}-a_{000}a_{110}) & a_{100}a_{011}+a_{101}a_{010}-a_{000}a_{111}-a_{001}a_{110} \\ a_{100}a_{011}+a_{101}a_{010}-a_{000}a_{111}-a_{001}a_{110} & 2(a_{101}a_{011}-a_{001}a_{111}) \end{array}\right) = 0
\eea
which when expanded gives the $2\times 2 \times2$ hyperdeterminant formula stated in equation (\ref{hyperdet_minor}) as expected.\\
(b) Conversely suppose that (\ref{new_hyper2}) vanishes and therefore there is a non-trivial solution for $W$ and $Z$ in (\ref{W_null}). To prove that there is also a non-trivial solution to (\ref{degenracy_conditions}), we need to show that such $X,~Y$ and $Z$ exist so that (\ref{null_cond}) and (\ref{eqn2}) hold. By definition of $w_0,w_1,w_2,w_3$, it is not hard to see that a valid $x_0, x_1, y_0$ and $y_1$ can be found from $w_i$ only if $W = (\begin{array}{cccc} w_0 & w_1 & w_2 & w_3 \end{array})^t$ in (\ref{W_null}) has the property,
\be \label{proportionW}
\frac{w_0}{w_2} = \frac{w_1}{w_3}
\ee
In the following we show that the solution of (\ref{W_null}) in fact satisfies relation (\ref{proportionW}). Let $p = \left(\begin{array}{cc} w_0 & w_1 \end{array} \right)^t$ and $q = \left(\begin{array}{cc} w_2 & w_3 \end{array}\right)^t$. Then from (\ref{W_null}) we obtain:
\bea
\alpha J q &+& {B_0}^t Z = 0\\
-\alpha J p &+& {B_1}^t Z = 0\\
\label{lasteqn} B_0 p &+& B_1 q = 0
\eea
Multiplying the first equation by $p^t$ and the second one by $q^t$ and adding them together we obtain,
\be
\alpha (p^t J q - q^t J p) + (p^t {B_0}^t + q^t {B_1}^t)Z = 0
\ee
which by the use of (\ref{lasteqn}) simplifies to:
\be
p^t J q = q^t J p
\ee
Noting that $p^t J q = (p^t J q)^t = - q^t J p$ gives,
\be \label{pqCond}
p^t J q = q^t J p = 0
\ee
(\ref{proportionW}) then follows immediately from (\ref{pqCond}) by substituting for $p$ and $q$.
\end{proof}

\subsection{Minors of a Symmetric Matrix Satisfy the Hyperdeterminant}

It has recently been shown in \cite{hyperdet}
that the principal minors of a symmetric matrix satisfy the
hyperdeterminant relations. There, this was found by explicitly computing the
determinant of a $3\times 3$ matrix in terms of the other minors and
noticing that it satisfied the $2\times 2\times 2$
hyperdeterminant. In this section we give an explanation of why this
relation holds for the principal minors of a symmetric matrix. The key
ingredient is identifying a simple determinant
formula for the multilinear form (\ref{multilinear}) when the coefficients $a_{i_1,i_2,\ldots,i_n}$ are the principal minors of an $n\times n$ symmetric matrix.

\begin{lemma}\label{lem:detformat}
Let the elements of the tensor $ [m_{i_1,i_2,\ldots,i_n}],~i_k \in \{0,1\}$ be the principal minors of an $n\times n$ matrix $\tilde{A}$ such that $m_{i_1,i_2,\ldots,i_n},~i_k \in \{0,1\}$ denotes the principal minor obtained by choosing the rows and columns of $\tilde{A}$ indexed by the set $\alpha = \{k | i_k = 1\}$ (by convention when all indices are zero $a_{00\ldots0} = 1)$. Then the following multilinear form of the format $2\times 2 \times \ldots \times 2$ ($n$ times),
\begin{equation}\label{original_form1}
f(X_{1},X_{2},\ldots,X_{n}) = \hspace{-5pt} \sum_{i_1,i_2,\ldots,i_n=0}^{1} \hspace{-10pt} m_{i_1,i_2,\ldots,i_n} x_{1,i_1} x_{2,i_2}\ldots x_{n,i_n}
\end{equation}
can be rewritten as the determinant of the matrix $F$, i.e., $f(X_{1},X_{2},\ldots,X_{n}) = \det( F )$ where $F$ is the following matrix:
\begin{eqnarray} \label{original_form2}
F & \hspace{-8pt} =& \hspace{-8pt}\left(\hspace{-4pt} \begin{array}{cccc} x_{1,0} &\hspace{-5pt} 0  &\hspace{-5pt} \ldots &\hspace{-5pt} 0 \\ 0 &\hspace{-5pt} x_{2,0} &\hspace{-5pt} \ldots &\hspace{-5pt} 0  \\ \vdots &  &\ddots &\\ 0 &\hspace{-5pt}  0  &\hspace{-5pt} \ldots &\hspace{-5pt} x_{n,0} \end{array}\hspace{-4pt} \right)
+ \left(\hspace{-4pt}\begin{array}{cccc} x_{1,1} &\hspace{-5pt} 0  &\hspace{-5pt} \ldots & \hspace{-5pt}0 \\ 0 &\hspace{-5pt} x_{2,1} &\hspace{-5pt} \ldots &\hspace{-5pt} 0  \\ \vdots &  &\ddots &\\ 0 &\hspace{-5pt}  0  &\hspace{-5pt} \ldots &\hspace{-5pt} x_{n,1} \end{array} \hspace{-4pt}\right) \tilde{A} 
\end{eqnarray}
\end{lemma}
\vspace{7pt}
\begin{proof}
First note that determinant of $F$ has the form, $\det(F) = \sum_{i_1,i_2,\ldots,i_n = 0}^{1} b_{i_1,i_2,\ldots,i_n} x_{1,i_1}x_{2,i_2}\ldots x_{n,i_n}$ for some $\tilde{A}$-dependent coefficients $b_{i_1,i_2,\ldots,i_n}$ (this is because each $x_{j,i_j}$ appears only in the $j$th row of $F$). To prove that $\det(F)$ is in fact equal to (\ref{original_form1}), we need to show that $b_{i_1,i_2,\ldots,i_n} = m_{i_1,i_2,\ldots,i_n},~\forall i_1,\ldots i_n$ or in other words $b_{i_1,i_2,\ldots,i_n}$ are the corresponding minors of $\tilde{A}$.\\
Let $(p_1\ldots p_n)$ be a realization of $\{0,1\}^n$. For $j=1,\ldots,n$, let the variables $x_{j,p_j}= 1$ and the rest of the variables be zero. This choice of values makes $\det(F) = b_{p_1,p_2,\ldots,p_n}$ and $f(X_1,X_2,\ldots,X_n) = m_{p_1,p_2,\ldots,p_n}$. Moreover it can be easily seen that in this case $\det(F)$ in (\ref{original_form2}) will simply be equal to the minor of the matrix $\tilde{A}$ obtained by choosing the set of rows and columns $\alpha \subseteq \{1,\ldots,n\}$ such that $p_j=1$ for all $j \in \alpha$. By assumption this is nothing but the coefficient $m_{p_1,p_2,\ldots,p_n}$ in (\ref{original_form1}) and therefore the lemma is proved.
\end{proof}
{\em Remark:} Note that Lemma \ref{lem:detformat} does not require the matrix $\tilde{A}$ to be symmetric.

\begin{lemma}[Partial derivatives of $\det F$]
Let $\alpha_j= \{1,\ldots,n\}\setminus j$ and $\alpha_k = \{1,\ldots,n\}\setminus k$. Computing the partial derivatives of $\det F$ gives:
\begin{eqnarray}
\label{partial1} &&\hspace{-0pt} \frac{\partial \det F}{\partial x_{j,\hspace{1pt} 0}} = \det F_{\alpha_j,\alpha_j}\\
\label{partial2} &&\hspace{-0pt}\frac{\partial \det F}{\partial x_{j,\hspace{1pt} 1}} = \sum_{k=1}^n \tilde{a}_{jk} (-1)^{j+k} \det F_{\alpha_j,\alpha_k}
\end{eqnarray}
where $\tilde{a}_{jk}$ denotes the $(j,k)$ entry of $\tilde{A}$ and $F_{\alpha_j,\alpha_k}$ denotes the submatrix of $F$ obtained by choosing the rows in $\alpha_j$ and columns in $\alpha_k$.
\end{lemma}
\begin{proof}
Consider the expansion of $\det F$ along its $j$th row as
$\det F = \sum_{k=1}^n F_{jk} (-1)^{j+k}\det F_{\alpha_j,\alpha_k}$. Noting that for $k\neq j, F_{jk} = x_{j,1}\tilde{a}_{jk}$ and for $j=k, F_{jj} = x_{j,0}+x_{j,1}\tilde{a}_{jj}$, we obtain that
\bea
\det F = x_{j,1} \sum_{k=1}^n \tilde{a}_{j,k} (-1)^{j+k} \det F_{\alpha_j,\alpha_k} + x_{j,0} \det F_{\alpha_j,\alpha_j}
\eea
Taking partial derivatives immediately gives (\ref{partial1}) and (\ref{partial2}).
\end{proof}


Now we can write the condition for the minors of $\tilde{A}$ to satisfy the hyperdeterminant:
\begin{lemma}[rank of $F$]\label{lem:mrank}
The principal minors of matrix $\tilde{A}$ satisfy the hyperdeterminant equation if there exists a set of solutions $x_{j,0}$ and $x_{j,1}$ for which rank of $F$ in (\ref{original_form2}) is at most $n-2$.
\end{lemma}
\begin{proof}
If there exists a nontrivial set of solutions $x_{j,0}$ and $x_{j,1}$ such that $F$ has rank $n-2$ then both (\ref{partial1}) and (\ref{partial2}) vanish (because all the $(n-1)\times (n-1)$ minors also vanish). But the vanishing of (\ref{partial1}) and (\ref{partial2}) simply means that the minors of $\tilde{A}$ satisfy the $2\times2 \times \dots \times 2$ ($n$ times) hyperdeterminant.
\end{proof}


\begin{theorem}[hyperdeterminant and the principal minors]
  The principal minors of an $n \times n$ symmetric matrix $\tilde{A}$ satisfy the hyperdeterminants of the format $2\times 2 \ldots \times 2$ ($k$ times) for all $k\leq n$.
\label{thm:principal_satisfy_hyper}
\end{theorem}
\begin{proof}
It is sufficient to show that the minors satisfy the $2\times 2 \ldots \times 2$ (n times) hyperdeterminant. Recall that for the tensor of coefficients $a_{i_1,i_2,\ldots,i_n}$ in the multilinear form (\ref{multilinear}) to satisfy the hyperdeterminant relation, there must exist a non-trivial solution to make all the partial derivatives of $f$ with respect to its variables zero. Lemma (\ref{lem:mrank}) suggests that a set of nontrivial $x_{j,0}$ and $x_{j,1}$ for which rank of $F$ is at most $n-2$ would be sufficient. 
In the following we will show that one can always find a solution to make $\mbox{rank}(F) \leq n-2$.\\
First we find a non-trivial solution in the case of 3 variables and then extend it to the the case where there are $n$ variables.
For 3 variables, the matrix $F$ which is of the following form,
\bea
F = \left(\begin{array}{ccc} x_{1,0}+x_{1,1}\tilde{a}_{11} & x_{1,1}\tilde{a}_{12} & x_{1,1}\tilde{a}_{13} \\ x_{2,1}\tilde{a}_{12} & x_{2,0}+x_{2,1}\tilde{a}_{22} & x_{2,1}\tilde{a}_{23} \\ x_{3,1}\tilde{a}_{13} & x_{3,1}\tilde{a}_{23} & x_{3,0}+x_{3,1}\tilde{a}_{33}
\end{array}\right)
\eea
should be rank 1 or equivalently all the rows be multiples of one another. Enforcing this condition results in 3 equations for 6 unknowns. Therefore without loss of generality we let $x_{j,1}=1$. Making the rows of $F$ proportional, gives:
\bea
\frac{x_{1,0}+\tilde{a}_{11}}{\tilde{a}_{12}} &\hspace{-3pt}=&\hspace{-3pt} \frac{\tilde{a}_{12}}{x_{2,0}+\tilde{a}_{22}}\hspace{5pt} = \hspace{5pt} \frac{\tilde{a}_{13}}{\tilde{a}_{23}}\\
\frac{\tilde{a}_{13}}{\tilde{a}_{12}} &\hspace{-3pt}=&\hspace{-3pt} \frac{\tilde{a}_{23}}{x_{2,0}+\tilde{a}_{22}} \hspace{5pt}= \hspace{5pt}\frac{x_{3,0}+\tilde{a}_{33}}{\tilde{a}_{23}} 
\eea
If $\overline{x}_j = (x_{j,0}, x_{j,1})$, then the solution to the above equations is clearly as follows:
\begin{eqnarray}\label{solnfor3}
\nonumber \overline{x}_1=(\frac{\tilde{a}_{12}\tilde{a}_{13}-\tilde{a}_{11}\tilde{a}_{23}}{\tilde{a}_{23}},1)\\
\nonumber \overline{x}_2=(\frac{\tilde{a}_{23}\tilde{a}_{12}-\tilde{a}_{13}\tilde{a}_{22}}{\tilde{a}_{13}},1)\\
\overline{x}_3=(\frac{\tilde{a}_{13}\tilde{a}_{23}-\tilde{a}_{12}\tilde{a}_{33}}{\tilde{a}_{12}},1)
\end{eqnarray}
Now for the general case of $n$ variables, let $\overline{x}_1,\overline{x}_2,\overline{x}_3$ be as (\ref{solnfor3}) and for $j>3,~~ \overline{x}_j = (1,0)$. It can be easily checked that this solution makes the matrix $F$ of rank $n-2$ and therefore the principal minors satisfy the $2\times 2\times \ldots \times 2$ ($n$ times) hyperdeterminant. \footnote{Note that the solutions (\ref{solnfor3}) also appear in \cite{hyperdet} in an alternative proof of principal minors satisfying the hyperdeterminant relation.}
\end{proof}

\begin{notation} \label{note:equiAa}
Each element $m_{i_1i_2,\dots,i_n}$, where $i_k \in \{0,1\}$, can alternatively be represented as $m_{\alpha}, \alpha\subseteq \{1,\dots,n\}$ where $\alpha = \{k| i_k = 1\}$. For example, $m_{100} = m_1$ and $m_{011} = m_{23}$.
\end{notation}

Since based on Theroem \ref{thm:principal_satisfy_hyper}, the principal minors of a symmetric matrix denoted by $m_{i_1i_2,\dots,i_n}$ satisfy the hyperdeterminant, we may write the $2\times2\times2$ hyperdeterminant relation of (\ref{hyperdet_minor}) for the principal minors of a $3\times 3$ matrix.
Adopting notation \ref{note:equiAa} we obtain:
\bea \label{hyper_principal_expand}
\nonumber m_{\emptyset}^2m_{123}^2 + m_1^2m_{23}^2 + m_2^2m_{13}^2 + m_3^2m_{12}^2 + 4m_{\emptyset}m_{12}m_{13}m_{23} + 4 m_1 m_2 m_3 m_{123} - 2m_{\emptyset}m_1 m_{23}m_{123}\\
- 2m_{\emptyset}m_2 m_{13}m_{123} - 2 m_{\emptyset}m_3 m_{12}m_{123} - 2m_1m_2m_{13}m_{23} - 2 m_1m_3m_{12}m_{23} - 2m_2m_3m_{12}m_{13} = 0
\eea
Letting $m_{\emptyset} = 1$, the $2\times 2 \times 2$ hyperdeterminant can also be written as,
\bea \label{hyper_principal_compact}
(m_{123} - m_3 m_{12} - m_2 m_{13} - m_1 m_{23} + 2m_1m_2m_3)^2 = 4(m_1m_2-m_{12})(m_1m_3-m_{13})(m_2m_3-m_{23}) 
\eea

\section{Minimal number of conditions for the elements of a $(2^n-1)-$dimensional vector to be the principal minors of a symmetric $n\times n$ matrix for $n \geq 4$} \label{Gaussian_sec:minimal}

In order to determine whether a $2^n-1$ dimensional vector $g$ corresponds to the entropy vector of $n$ scalar jointly Gaussian random variables, one needs to check whether the vector $e^g$ corresponds to all the principal minors of a symmetric positive semi-definite matrix.
Define $A \triangleq e^g$ and let the elements of the vector $A \in \mathds{R}^{2^n-1}$ be denoted by $A_\alpha,~\alpha \subseteq \{1, \ldots, n \}, \alpha\neq \emptyset$. An interesting problem is to find the minimal set of conditions under which the vector $A$ can be considered as the vector of all principal minors of a symmetric $n \times n$ matrix. This problem is known as the ``principal minor assignment'' problem and has been addressed before in \cite{hyperdet,minor-assign}. In fact in a recent remarkable work, \cite{hyperdet} gives the set of necessary and sufficient conditions for this problem. Nonetheless it does not point out the minimal set of such necessary and sufficient equations. Instead \cite{hyperdet} is mainly interested in the generators of the prime ideal of all homogenous polynomial relations among the principal minors of an $n\times n$ symmetric matrix. Here we propose the minimal set of such conditions for $n\geq 4$.

Roughly speaking there are $2^n-1$ variables in the vector $A$ and only $\frac{n(n+1)}{2}$ parameters in a symmetric $n \times n$ matrix. Therefore if the elements of $A$ can be considered as the minors of a $n\times n$ symmetric matrix, one suspects that there should be $2^n-1-\frac{n(n+1)}{2}$ constraints on the elements of $A$. These constraints which can be translated to relations between the elements of the entropy vector arising from $n$ scalar Gaussian random variables, can be used as the starting point to determining the entropy region of $n \geq 4$ jointly
Gaussian scalar random variables.

We start this section by studying the entropy region of 4 jointly Gaussian random variables using the results of the hyperdeterminant already mentioned in the previous section and we shall explicitly state the sufficiency of 5 constraints among all the constraints given in \cite{hyperdet} by using a similar proof to \cite{hyperdet}; that for a given vector $A$, and under such constraints, one can construct the symmetric matrix $\tilde{A} = [\tilde{a}_{ij}]$ with the desired principle minors.
Later in this section we state such minimal number of conditions for a $2^n-1$ dimensional vector for $n\geq 4$.
Now define
\begin{equation}
\label{eqn:gA}
c_{ijk} \triangleq A_{ijk} - A_i A_{jk} - A_{j}A_{ik} - A_k A_{ij} + 2A_iA_jA_k
\end{equation}
\begin{theorem} \label{thm:minimal-nece-suff} 
Let $A$ be a 15 dimensional vector whose elements are indexed by non-empty subsets of $\{1,2,3,4\}$ and has the property that 
\be
\label{eqn:nondegen_4}
A_{ij}<A_iA_j,~~~\forall i,j \subseteq \{1,2,3,4\}.
\ee 
Then the minimal set of necessary and sufficient conditions for the elements of the vector $A$ to be the principal minors of a symmetric $4\times 4$ matrix consists of three hyperdeterminant equations (\ref{hyper1}--\ref{hyper3}), one consistency of the signs of $c_{ijk}$ (\ref{gconsistency}), and the determinant identity of the $4\times 4$ matrix (\ref{det4}):
\bea
&&\hspace{-20pt}\label{hyper1}c_{123}^2 = 4 (A_1 A_2 - A_{12})(A_2 A_3 - A_{23})(A_1 A_3 - A_{13})\\
&&\hspace{-20pt}\label{hyper2}c_{124}^2 = 4 (A_1 A_2 - A_{12})(A_2 A_4 - A_{24})(A_1 A_4 - A_{14})\\
&&\hspace{-20pt}\label{hyper3}c_{134}^2 = 4 (A_1 A_3 - A_{13})(A_3 A_4 - A_{34})(A_1 A_4 - A_{14})\\
&&\hspace{-20pt} c_{123}c_{124}c_{134}=  4 (A_1 A_2 -A_{12})(A_1A_3-A_{13})(A_1A_4-A_{14}) c_{234}  \label{gconsistency}\\
&&\hspace{-20pt}\nonumber A_{1234} = -\frac{1}{2} \hspace{-15pt}\sum_{ \substack { i',j'\in \{1,2,3\} \\ k',l' \in \{1,2,3,4\} \setminus \{i',j'\}}}\hspace{-15pt} \frac{c_{i'j'k'}c_{i'j'l'}}{A_{i'}A_{j'} - A_{i'j'}} + A_1 c_{234}+ A_2 c_{134}\\
&& \hspace{20pt} + A_3 c_{124} + A_4 c_{123}- 2 A_1 A_2 A_3 A_4+ A_{12}A_{34} + A_{13} A_{24} + A_{14}A_{23} \label{det4} 
\eea
\end{theorem}
\vspace{10pt}
\begin{proof}
a) It is easy to show the necessity of equations (\ref{hyper1})--(\ref{det4}) and it was done in \cite{hyperdet}. Here we illustrate the method to make the paper self-contained. Note that if elements of the vector $A$ are the principal minors of a symmetric matrix then by Theorem \ref{thm:principal_satisfy_hyper} they satisfy the hyperdeterminant relations which from (\ref{hyper_principal_compact}) can be written as,
\bea \label{hyperA}
(A_{123}-A_3A_{12}-A_2A_{13}-A_1A_{23}+2A_1A_2A_3)^2 = 4(A_1A_2-A_{12})(A_1A_3-A_{13})(A_2A_3-A_{23})
\eea
Using the definition of $c_{ijk}$ in (\ref{eqn:gA}), equation (\ref{hyperA}) can be further written as:
\bea
c_{123}^2 = 4(A_1A_2-A_{12})(A_1A_3-A_{13})(A_2A_3-A_{23})
\eea 
Therefore equations (\ref{hyper1})--(\ref{hyper3}) simply represent the hyperdeterminant relations and hold by Theorem \ref{thm:principal_satisfy_hyper}. Moreover if $A_{ijk}$ is the principal minor of matrix $\tilde{A}$ obtained by choosing rows and columns $i$,$j$ and $k$, then writing $A_{ijk}$ in terms of the entries of $\tilde{A}$ gives,
\bea
\nonumber  A_{ijk} &  =& \tilde{a}_{ii}\tilde{a}_{jj}a_{kk} - \tilde{a}_{ii}\tilde{a}_{jk}^2 - \tilde{a}_{jj}\tilde{a}_{ik}^2 - \tilde{a}_{kk}\tilde{a}_{ij}^2 + 2 \tilde{a}_{ij}\tilde{a}_{jk}\tilde{a}_{ik}\\
 &=& -2A_iA_jA_k + A_i A_{jk} + A_{j}A_{ik} + A_k A_{ij} + 2\tilde{a}_{ij}\tilde{a}_{jk}\tilde{a}_{ik} \label{expand-det}
\eea
where since $A$ corresponds to principal minors of $\tilde{A}$, we have substituted for $\tilde{a}_{ii} = A_i$ and $\tilde{a}_{ij}^2 = A_iA_j - A_{ij}$. Rewriting (\ref{expand-det}) we obtain,
\bea
A_{ijk} - A_i A_{jk} - A_jA_{ik} - A_kA_{ij} + 2A_iA_jA_k = 2\tilde{a}_{ij}\tilde{a}_{jk}\tilde{a}_{ik}
\eea
which by comparison to (\ref{eqn:gA}) means that the following holds,
\bea
\label{gijk-atilde}
c_{ijk} = 2\tilde{a}_{ij}\tilde{a}_{jk}\tilde{a}_{ik}
\eea 
Therefore replacing for $c_{123}, c_{124}, c_{134}$ and $c_{234}$ from (\ref{gijk-atilde}) in (\ref{gconsistency}) and simplifying we obtain that (\ref{gconsistency}) holds trivially. The last condition (\ref{det4}) is also nothing but the expansion of the $4\times 4$ determinant of $\tilde{A}$ in terms of the entries of $\tilde{A}$ and replacing for them in terms of $c_{ijk}$ and lower order minors and therfore is a necessary condition.

b) Now we need to show the sufficiency of equations (\ref{hyper1})--(\ref{det4}). To do so, we assume that the given vector $A$ satisfies (\ref{hyper1})--(\ref{det4}) and we want to show that it is the principal minor vector of some symmetric matrix $\tilde{A}$. Hence we try to construct a matrix $\tilde{A}$ whose principal minors are given by the entries of vector $A$. First note that such a matrix $\tilde{A}$ should have $\tilde{a}_{ii} = A_i$ and $\tilde{a}_{ij}^2 = A_iA_j - A_{ij}$ (or equivalently $\tilde{a}_{ij} = \pm \sqrt{A_iA_j-A_{ij}}$). The only ambiguity in fully determining the entries of $\tilde{A}$ will therefore be the signs of the off-diagonal entries. 
To have the $3\times 3$ minors of $\tilde{A}$ also equal to the corresponding entries of vector $A$, we should have:
\bea
A_{ijk} &= &\det \left(\begin{array}{ccc} \tilde{a}_{ii} & \tilde{a}_{ij} & \tilde{a}_{ik} \\ \tilde{a}_{ij} & \tilde{a}_{jj} & \tilde{a}_{jk}\\ 
\tilde{a}_{ik} & \tilde{a}_{jk} & \tilde{a}_{kk}  \end{array}\right)  \\
&=&\tilde{a}_{ii}\tilde{a}_{jj}\tilde{a}_{kk} - \tilde{a}_{ii}\tilde{a}_{jk}^2 - \tilde{a}_{jj}\tilde{a}_{ik}^2 - \tilde{a}_{kk}\tilde{a}_{ij}^2 + 2 \tilde{a}_{ij}\tilde{a}_{jk}\tilde{a}_{ik}
\eea  
Replacing for values of $\tilde{a}_{ii}$ and $\tilde{a}_{ij}$ in terms of $A_i$ and $A_{ij}$ we obtain that,
\bea
A_{ijk} - A_iA_{jk} - A_j A_{ik} - A_kA_{ij} + 2A_iA_jA_k = 2\tilde{a}_{ij}\tilde{a}_{jk}\tilde{a}_{ik}
\eea
Therefore writing the condition for all $\{i,j,k\} \subseteq \{1,2,3,4\}$ we need to have
\bea
c_{123} = 2\tilde{a}_{12}\tilde{a}_{23}\tilde{a}_{13} \label{eqn:g123}\\
c_{124} = 2\tilde{a}_{12}\tilde{a}_{14}\tilde{a}_{24} \label{eqn:g124}\\
c_{134} = 2\tilde{a}_{13}\tilde{a}_{14}\tilde{a}_{34} \label{eqn:g134}\\
c_{234} = 2\tilde{a}_{23}\tilde{a}_{24}\tilde{a}_{34} \label{eqn:g234}
\eea
Note that the constraints (\ref{hyper1})--(\ref{gconsistency}) guarantee that $|c_{ijk}| = 2|\tilde{a}_{ij}\tilde{a}_{ik}\tilde{a}_{jk}|$. It remains to show that there is a consistent choice of signs for $\tilde{a}_{ij}$ such that the stronger condition $c_{ijk} = 2\tilde{a}_{ij}\tilde{a}_{jk}\tilde{a}_{jk}$ also holds. Note that without loss of generality we can assume that the off-diagonal entries on the first row, i.e., $\tilde{a}_{1j}$ are positive. This is due to the fact that we can use the transformation $D\tilde{A}D^{-1}$ where $D$ is a diagonal matrix with $\pm 1$ elements to make $\tilde{a}_{1j}$ positive without affecting any of the principal minors. Hence, assuming $\tilde{a}_{1j}$ are positive, the signs of $\tilde{a}_{23}, \tilde{a}_{24}$ and $\tilde{a}_{34}$ are determined from the signs of $c_{123}, c_{124}$ and $c_{134}$ in (\ref{eqn:g123}--\ref{eqn:g134}). However once the signs of $\tilde{a}_{23}, \tilde{a}_{24}$ and $\tilde{a}_{34}$ are determined, the sign of their product, i.e., $\tilde{a}_{23}\tilde{a}_{24}\tilde{a}_{34}$ should be the same as the sign of $c_{234}$ to satisfy (\ref{eqn:g234}). This is enforced by equation (\ref{gconsistency}). Therefore conditions (\ref{hyper1})--(\ref{gconsistency}) yield (\ref{eqn:g123})--(\ref{eqn:g234}). 
Note that due to property (\ref{eqn:nondegen_4}) none of the $\tilde{a}_{ij}$ are zero and hence all the above steps for sign choice are valid.
Finally a direct calculation of the $4\times 4$ determinant of $\tilde{A}$ shows that it can be expressed as the right-hand side of equation (\ref{det4}) in terms of lower order minors of $\tilde{A}$ which are equal to corresponding terms $A_\alpha, |\alpha|\leq 3$. Consequently equation (\ref{det4}) guarantees that the $4\times 4$ principal minor of $\tilde{A}$ is equal to $A_{1234}$. Again note that since property (\ref{eqn:nondegen_4}) holds, the denominator in (\ref{det4}) is nonzero and will not cause any problems. 
Therefore $\tilde{A}$ will be the matrix with principal minors given by vector $A$. 
%
\end{proof}


Using a similar approach which follows the proof methods of \cite{hyperdet} closely, we can write the set of minimal necessary and sufficient conditions for a $2^n-1$ dimensional vector to be the principal minors of a symmetric matrix.
\begin{theorem} \label{thm:minimal-necesuff-n}  
Let $A$ be a $2^n-1$ dimensional vector whose elements are indexed by non-empty subsets of $\{1,\dots,n\}$ (assume $A_{\emptyset} = 1$) and that it satisfies,
\be
\label{eqn:nondegen_n}
A_{\alpha}A_{i\cup j\cup\alpha} < A_{i\cup \alpha}A_{j\cup \alpha}~~~\forall i,j\in \{1,\dots,n\}~\alpha \subseteq \{1,\dots,n\}\setminus \{i,j\}
\ee
Then the necessary and sufficient conditions for a $2^n-1$ dimensional vector to be the principal minors of a symmetric $n\times n$ matrix consists of $2^n-1-\frac{n(n+1)}{2}$ equations and are as follows,
\bea
&& \forall j,k \subseteq\{2,\ldots,n\}~~c_{1jk}^2 = 4(A_1A_j-A_{1j})(A_1A_k-A_{1k})(A_jA_k-A_{jk}) \label{nece-suff1}\\
&& \forall i,j,k \subseteq\{2,\ldots,n\}~~c_{1ij}c_{1ik}c_{1jk} = 4(A_1A_i-A_{1i})(A_1A_j-A_{1j})(A_1A_k-A_{1k})c_{ijk} \label{nece-suff2}
\eea
Also $\forall \beta \subseteq\{1,\ldots,n\},|\beta|\geq4$ choose {\it one} set of $i,j,k,l \subseteq \beta$ such that $i<j<k<l$ and let $\alpha = \beta\backslash \{i,j,k,l\}$,
\bea
D_{ijkl}^\alpha = 0 \label{nece-suff3}
\eea
where $D_{ijkl}^\alpha$ is obtained from the following by replacing every $A_S, S\subseteq \{i,j,k,l\}$ by $\frac{A_{S \cup\alpha}}{A_\alpha}$.
\bea
\nonumber && \hspace{-20pt} D_{ijkl} = A_{ijkl} + \frac{1}{2} \hspace{-15pt}\sum_{ \substack { i',j'\in \{i,j,k\} \\ k',l' \in \{i,j,k,l\} \setminus \{i',j'\}}}\hspace{-15pt} \frac{c_{i'j'k'}c_{i'j'l'}}{A_{i'}A_{j'} - A_{i'j'}} - A_i c_{jkl}- A_j c_{ikl}\\
\nonumber && \hspace{20pt} - A_k c_{ijl} - A_l c_{ijk}+ 2 A_i A_j A_k A_l- A_{ij}A_{kl} \\
&& \label{det4-ijkl}  \hspace{20pt} - A_{ik} A_{jl} - A_{il}A_{jk} = 0
\eea
\end{theorem}
\begin{proof}
The proof is essentially the same as the proof technique of \cite{hyperdet} and is a generalization of Theorem \ref{thm:minimal-nece-suff} to a $2^n-1$ dimensional vector. However we would like to highlight why (\ref{nece-suff1})--(\ref{nece-suff3}) is the {\it minimal} set of necessary and sufficient conditions among all conditions given in \cite{hyperdet}. First consider the necessitiy of the equations:

a) Showing the necessity of equations (\ref{nece-suff1})--(\ref{nece-suff3}) is strightforward. In fact if the elements of vector $A$ are the principal minors of a symmetric matrix $\tilde{A}$, then $\tilde{a}_{ii} = A_i$ and $\tilde{a}_{ij}^2 = A_iA_j-A_{ij}$.
Furthermore it can be shown (similar to the proof of Theorem \ref{thm:minimal-nece-suff}) that 
$c_{ijk} = 2\tilde{a}_{ij}\tilde{a}_{ik}\tilde{a}_{jk}$ from which it follows that (\ref{nece-suff1})--(\ref{nece-suff2}) hold. Note that $D_{ijkl} = 0$ in equation (\ref{det4-ijkl}) gives $A_{ijkl}$ in terms of lower order minors (compare to equation (\ref{det4})). Now consider a submatrix of $\tilde{A}$ whose rows and column are indexed by $\beta\subseteq \{1,\dots,n\},~|\beta|\geq 4$ and denote it by $\tilde{A}_{[\beta]}$. For $\{i,j,k,l\}\subseteq \beta$ let $\alpha = \beta  \setminus \{i,j,k,l\}$ and likewise let the submatrix with rows and columns indexed by $\alpha$ be shown by $\tilde{A}_{[\alpha]}$. Further denote the Schur complement of $\tilde{A}_{[\alpha]}$ in $\tilde{A}_{[\beta]}$ by $\tilde{A}_{[\alpha/\beta]}$. Note that $\tilde{A}_{[\alpha/\beta]}$ is a $4\times 4$ matrix whose determinant can be obtained via the rule of equation (\ref{det4-ijkl}). The property of Schur complement yields, \bea
\label{eqn:schurprop}
\det \left(\tilde{A}_{[\alpha/\beta]}\right)_{[S]}  = \frac{\det \tilde{A}_{[S\cup \alpha]}}{\det \tilde{A}_{[\alpha]} } = \frac{A_{S\cup \alpha}}{A_\alpha}  ,~\forall S\subseteq \{i,j,k,l\}
\eea
Therefore wiritng the determinant for the $4\times 4$ matrix $\tilde{A}_{[\alpha/\beta]}$ and using (\ref{eqn:schurprop}), gives equation (\ref{nece-suff3}). \footnote{Note that we can assume $g_\alpha \neq -\infty$ and simply avoid $A_{\alpha}= 0$.}

b) To show the sufficiency we show that if a given vector $A$ satisfies equations (\ref{nece-suff1})--(\ref{nece-suff3}) then we can construct a symmetric matrix whose principal minors are given by $A$. As we did in Theorem \ref{thm:minimal-nece-suff} such a matrix $\tilde{A}$ should have entries $\tilde{a}_{ii} = A_{i}$ and $\tilde{a}_{ij}^2 = A_iA_j - A_{ij}$ (or equivalently $\tilde{a}_{ij} = \pm \sqrt{A_iA_j-A_{ij}}$). Therefore it remains to choose the signs of the off-diagonal entries in a consistent fashion so that all the minors of $\tilde{A}$ will correspond to $A$.

For $3\times 3$ minors of $\tilde{A}$ to comply with $A_{ijk}$, note that as was obtained in Theorem \ref{thm:minimal-nece-suff}, we need to have $\forall \{i,j,k\} \subseteq \{1,\dots,n\},~ 2\tilde{a}_{ij}\tilde{a}_{jk}\tilde{a}_{ik} = c_{ijk}$. Note that equations (\ref{nece-suff1})--(\ref{nece-suff3}) give that $|c_{ijk}| = 2|\tilde{a}_{ij}\tilde{a}_{jk}\tilde{a}_{ik}|$ for all $\{i,j,k\}\subseteq \{1,\dots,n\}$. Again similar to Thereom \ref{thm:minimal-nece-suff}, we may assume that all the off-diagonal terms in the first row are positive since we can always use the transformation $D\tilde{A}D^{-1}$ where $D$ is a diagonal matrix with $\pm 1$ elements to make $\tilde{a}_{1j}$ positive. Therefore assuming $\tilde{a}_{1j}$ are positive, we can choose the sign of all off-diagonal terms $\tilde{a}_{jk}, ~j,k\in \{2,\dots,n\}$ such that they have the same sign as $c_{1jk}$. This way we guarantee $2\tilde{a}_{1j} \tilde{a}_{1k} \tilde{a}_{jk} = c_{1jk}$. However note that all the signs of all entries of $\tilde{A}$ are now fixed and therefore we need to make sure the remaining conditions $c_{ijk} = 2\tilde{a}_{ij} \tilde{a}_{ik} \tilde{a}_{jk}$ for all $\{i,j,k\}\subseteq \{2,\dots ,n\}$ are also satisfied. However since $c_{1jk} = 2\tilde{a}_{1j} \tilde{a}_{1k} \tilde{a}_{jk}$ and $A$ satisfies the constraint (\ref{nece-suff2}) as well, $c_{ijk} = 2\tilde{a}_{ij} \tilde{a}_{ik} \tilde{a}_{jk}$ for $\{i,j,k\} \subseteq \{2,\dots,n\}$ follows immediately. Therefore equations (\ref{nece-suff1}) and (\ref{nece-suff2}) guarantee the equality of $3\times 3$ minors of $\tilde{A}$ with the corresponding entries of $A$. Note that property (\ref{eqn:nondegen_n}) assures that none of the $\tilde{a}_{ij}$ are zero and hence all the above steps for determining the sign will be valid.

Now we need to prove the equality of all minors of size $4\times 4$ and bigger of $\tilde{A}$ with the relative entries of $A$. This is enforced by condition (\ref{nece-suff3}). 
To see the reason, replace each term $A_\gamma,~\gamma\subseteq \{1,\dots,n\}$ in (\ref{nece-suff3}) by $\det \tilde{A}_{[\gamma]}$. Then as we saw in part (a), the resulting equation describes $\det \tilde{A}_{[\{ijkl\}\cup \alpha]} = \det \tilde{A}_{[\beta]}$ in terms of lower order minors which are already guaranteed to be equal to the corresponding entries of $A$. 
Therefore condition (\ref{nece-suff3}) is enforcing $\det \tilde{A}_{[\beta]} = A_\beta$. 
Note that for each $\beta$ only 1 equation of type (\ref{nece-suff3}) is required. Moreover due to property (\ref{eqn:nondegen_n}) the denominators in (\ref{nece-suff3}) obtained from (\ref{det4-ijkl}) will be nonzero and will not cause any problems. 

Finally note that there are ${n-1 \choose 2}$ number of constraints of type (\ref{nece-suff1}), ${n-1 \choose 3}$ of type (\ref{nece-suff2}) and $\sum_{m = 4}^n {n \choose m}$ of type (\ref{nece-suff3}) which sums up to $2^n - 1 - \frac{n(n+1)}{2}$. This is the number that we expect noting that there are only $\frac{n(n+1)}{2}$ free parameters in a symmetric matrix while the vector of principal minors is of size $2^n-1$.
\end{proof}

\begin{corollary}
In Theorem \ref{thm:minimal-necesuff-n} if we insist that for all $\alpha \subseteq \{1,\ldots,n\},~A_\alpha\geq 0$ and substitute each $A_\alpha$ by $e^{g_\alpha}$ in (\ref{nece-suff1})-(\ref{nece-suff3}) then (\ref{nece-suff1})-(\ref{nece-suff3}) give the necessary and sufficient conditions for a $2^n-1$ dimensional vector to correspond to the entropies of $n$ scalar jointly Gaussian random variables. 
\end{corollary}

{\em Remark:} Note that in order to characterize the entropy region of scalar Gaussian random variables what one really needs is the {\it convex hull} of all such entropy vectors. After all if we only wanted to determine whether 7 numbers correspond to the entropy vector of 3 scalar-valued jointly Gaussian random variables we could simply check whether they satisfy the hyperdeterminant relation (\ref{hyper_principal_compact}). However it is the convex hull which is more interesting, and more cumbersome to calculate, and this is what we addressed for 3 random variables in Section \ref{Gaussian_sec:main3}. 

\section{Discussions and Conclusions} \label{Gaussian_sec:conclusion} 

In this paper, we studied the entropy region of jointly Gaussian random variables as an interesting subclass of continuous random variables. In particular we determined that the whole entropy region of 3 arbitrary continuous random variables can be obtained from the convex cone of the entropy region of 3 scalar-valued jointly Gaussian random variables. We also gave the representation of the entropy region of 3 vector-valued jointly Gaussian random variables through a conjecture. 

We should remark that, in general, to characterize the entropy region of Gaussian random variables one should consider vector-valued random variables which is probably more complex than the case of scalars. In Section \ref{Gaussian_sec:main3} we showed that, for $n=3$, the vector-valued random variables do not result in a bigger region than the convex hull of scalar ones.
However in general it is not known whether the entropy region of $n$ vector-valued jointly Gaussian random variables is greater than the convex hull of the entropy region of scalar valued Gaussians.

For $n\geq 4$ we explicitly stated the set of $2^n-1 - \frac{n(n+1)}{2}$ constraints that a $2^n-1$ dimensional vector should satisfy in order to correspond to the entropy vector (equivalently the vector of all principal minors) of $n$ {\it scalar} jointly Gaussian random variables. Although these conditions allow one to check whether a real vector of $2^n-1$ numbers corresponds to an entropy vector of $n$ scalar jointly Gaussian random variables, they do not reveal if such given vector of $2^n-1$ real numbers corresponds to the entropy vector of vector-valued Gaussian random variables or if it lies in the {\it convex hull} of scalar Gaussian entropy vectors. Answering these question requires one to study the region of {\it vector valued} jointly Gaussian random variables and this is what we addressed in Section III for 3 random variables. Obtaining the entropy region of vector-valued Gaussians seems to be rather complicated for $n\geq 4$ and as a satrting point one may instead focus on the convex hull of scalar Gaussians which is essentially the convex hull of vectors satisfying constraints (\ref{nece-suff1})--(\ref{nece-suff3}).   
Studying such a convex hull has an interesting connection to the concept of an ``amoeba'' in algebraic geometry. 
The ``amoeba'' of a polynomial $f(x_1,\ldots,x_k) = \sum_i q_i x_1^{p_{1i}}\ldots x_k^{p_{ki}}$ is defined as the image of $f = 0$ in $\mathds{R}^k$ under the mapping $(x_1,\ldots,x_k)\mapsto (\log |x_1|,\ldots,\log |x_k|)$ \cite{gelfand}. It turns out that many properties of amoebas can be deduced from the Newton polytope of $f$ which is defined as the convex hull of the exponent vectors $(p_{1i},\ldots,p_{ki})$ in $\mathds{R}^k$ (see, e.g., \cite{PR-amoeba-newton}). In terms of our problem of interest, the scalar Gaussian entropy points are the intersection of the amoebas associated to polynomials (\ref{nece-suff1})-(\ref{nece-suff3}) and one should look for the convex hull of the locus of these intersection points. If we allow the notion of amoeba to be defined as the $\log$ mapping for any function (not just polynomials), then one could also formulate our problem of interest as the convex hull of the amoeba of the algebraic variety obtained from the intersection of (\ref{nece-suff1})-(\ref{nece-suff3}).


Finally in characterizing the entropy region of Gaussian random variables for $n\geq 4$, we noted the important role of the  hyperdeterminant and 
with this viewpoint we also examined the hyperdeterminant relations. In particular by giving a determinant formula for a multilinear form, we gave a transparent proof that the hyperdeterminant relation is satisfied by the principal minors of an $n\times n$ symmetric matrix. Moreover we also obtained a determinant form for the $2\times 2\times 2$ hyperdeterminant which might be extendible to higher order formats and is an interesting problem even in its own right.


\newpage

\section*{Appendix\\ Proof of Lemma \ref{lem:f}}
\label{sec:App}

{\em Proof:} 
We will first show that $\forall i,j,~f(\delta) \leq x_i x_j$. Let,
\bea
e(\delta) = -2+\sum_{l=1}^3x_l^{{\delta}}+ 2\sqrt{\prod_{l=1}^3(1-x_l^{{\delta}})}
\eea
For distinct $i,j,k \subseteq \{1,2,3\}$, this can also be written as,
\bea
e(\delta) &=& (x_i x_j)^{{\delta}} - \left( (1-x_i^{{\delta}})(1-x_j^{\delta}) + (1-x_k^{{\delta}}) - 2 \sqrt{(1-x_i^{{\delta}})(1-x_j^{{\delta}})(1-x_k^{{\delta}})} \right)\\
&=& (x_i x_j)^{{\delta}} - \left(\sqrt{(1-x_i^{{\delta}})(1-x_j^{{\delta}})}-\sqrt{1-x_k^{{\delta}}} \right)^2
\eea
which shows $e(\delta) \leq (x_i x_j)^{{\delta}}$ and therefore for all $\delta \geq 0$, $f(\delta) \leq x_ix_j$  with equality if and only if $(1-x_i^{{\delta}})(1-x_j^{{\delta}}) = 1- x_k ^{\delta}$ or equivalently,
\be
(x_i x_j)^{{\delta}} + x_k^{{\delta}} - x_i^{{\delta}}- x_j^{{\delta}} = 0
\ee

Note that this is only possible when $x_ix_j = \frac{x_1x_2 x_3}{\max_l x_l}$. Without loss of generality assume, $x_1\leq x_2\leq x_3$, and define,
\be
y(\delta) = (x_1 x_2)^{\delta} + x_3^{{\delta}} - x_1^{{\delta}} - x_2^{{\delta}}
\ee
Clearly zeros of $y({\delta})$ determine the global maximums of $f(\delta)$ (i.e., $f(\delta) = x_1x_2$). Therefore we analyze the behavior of $y(\delta)$ in the following scenarios (based on the assumption $x_1\leq x_2 \leq x_3$):
\begin{enumerate}
\item $x_1<x_2<x_3<1$:\\
Note that for any number $0<x<1$, when $\delta\rightarrow 0$, we have the approximation, $x^{\delta} \approx 1+\delta \log x $. Therefore as $\delta \rightarrow 0^+$, we obtain,
\be
y(\delta) \approx \delta \log x_3 < 0 
\ee 
On the other hand when $\delta \rightarrow \infty$, we have $y(\delta) \approx x_3 ^{\delta}  > 0$. Therefore $y(\delta)$ has at least one zero (which is not at origin). In fact we now show that it has exactly one zero. Let $a(\delta) = \frac{y(\delta)}{x_3^{\delta}} = 1 + (\frac{x_1x_2}{x_3})^{\delta} -(\frac{x_1}{x_3})^{\delta} - (\frac{x_2}{x_3})^{\delta}$. Then zeros of $y(\delta)$ and $a(\delta)$ match except possibly at infinity. Therefore we can equivalently determine the zeros of $a(\delta)$. To do so, we further define $b(\delta) = (\frac{x_3}{x_2})^{\delta}\cdot \frac{d a(\delta)}{d\delta} = x_1^{\delta}\log(\frac{x_1x_2}{x_3}) - (\frac{x_1}{x_2})^{\delta} \log(\frac{x_1}{x_3}) - \log(\frac{x_2}{x_3})$. Note that as $\delta\rightarrow 0^{+}$, $b(0^{+}) \approx \log x_3 <0$. Moreover as $\delta \rightarrow \infty$ we obtain $b(\infty) \approx \log(\frac{x_3}{x_2}) >0$. On the other hand obtaining the derivative of $b({\delta})$ gives,
\be
\frac{d b(\delta)}{d\delta} =  {x_1}^{\delta}\log (x_1) \log \left(\frac{x_1x_2}{x_3}\right) - \left(\frac{x_1}{x_2}\right)^{\delta}\log\left(\frac{x_1}{x_2}\right) \log\left(\frac{x_1}{x_3}\right) 
\ee
which has a unique zero $\delta^*$ at,
\be
x_2^{\delta^*} = \frac{\log\left(\frac{x_1}{x_2}\right) \log\left(\frac{x_1}{x_3}\right) }{ \log (x_1) \log \left(\frac{x_1x_2}{x_3}\right) }
\ee
Calculating $\frac{d^2 b(\delta^*)}{d\delta^2}$ shows that $b(\delta)$ has a maximum at $\delta^*$. 
Noting that derivative of $b(\delta)$ is defined for all $\delta$ yet it is only zero at $\delta^*$ and that $b(0^+) < 0 $ and $b(\infty)>0$, we conclude that $b(\delta)$ has exactly one zero at some $0<\hat{\delta}<\delta^*$.  Since we defined $b(\delta) =  (\frac{x_3}{x_2})^{\delta}\cdot \frac{d a(\delta)}{d\delta}$, we obtain that $\frac{d a(\delta)}{d\delta}$ has also exactly one zero at $\hat{\delta}$ (and is also possibly zero at infinity). Moreover we have $\delta\rightarrow 0^+$, $a(0^+) \approx \delta\log x_3 <0$, $\frac{da}{d\delta}(0^+) \approx \log x_3 < 0$ and $a(\infty) \approx 1$. Since $a(\delta)$ is also everywhere differentiable we deduce that $a(\delta)$ has exactly one zero (that is not at origin) at some ${\delta_0}$ where ${\delta_0}>\hat{\delta}>0$. Finally going back to $y(\delta)$, we obtain that $y(\delta)$ starts at origin, i.e., $y(0) = 0$ however with a negative slope $\frac{dy}{d\delta}(0^+) = \log x_3 <0 $. It has a unique zero at $\delta_0 >0$ and it approaches the $\delta$-axis again at infinity with a positive sign, i.e., $y(\delta) \approx x_3^\delta \rightarrow 0^+$ as $\delta\rightarrow \infty$. Therefore we have shown that $y(\delta)$'s behavior is as the one depicted in Fig. \ref{Gaussian_fig:fydelta}(a). 

Note that since the zeros of $y(\delta)$ determine the global maximums of $f(\delta)$, to determinne where the global maximums of $f(\delta)$ occur, we need to calculate $f(\delta)$ at $0, \infty$ and $\delta_0$. First for $\delta\rightarrow 0^+$ we have,
\be
f(\delta) \approx (1+\delta \log (x_1x_2x_3)) ^{\frac{1}{\delta}}\rightarrow x_1x_2x_3
\ee
Moreover for $\delta \rightarrow \infty$ we have $(1-x^{\delta})^{\frac{1}{2}} \approx 1-\frac{1}{2}x^\delta-\frac{1}{8}x^{2\delta}$ for $0<x<1$. Therefore we can write,
\bea
e(\delta) &\approx& -2+\sum_l x_l^{\delta} +2\left(1-\frac{1}{2}x_1^\delta-\frac{1}{8}x_1^{2\delta}\right)\left(1-\frac{1}{2}x_2^\delta-\frac{1}{8}x_2^{2\delta}\right)\left(1-\frac{1}{2}x_3^\delta-\frac{1}{8}x_3^{2\delta}\right)\\
&\approx& \frac{1}{2}(x_1x_2)^\delta +\frac{1}{2}(x_1x_3)^\delta + \frac{1}{2}(x_2x_3)^\delta - \frac{1}{4} x_1^{2\delta} -\frac{1}{4} x_2^{2\delta} -\frac{1}{4} x_3^{2\delta}   
\label{eqn:edeltataylor}
\eea
Since $\delta\rightarrow \infty$ and $x_1<x_2<x_3$, the term $x_3^{2\delta}$ will be the dominant term and we will have,
\be
e(\delta) \approx -\frac{1}{4} x_3^{2\delta} <0
\ee
Therefore as $\delta\rightarrow \infty$ we obtain $f(\delta) = (~max (0,e(\delta))~)^{\frac{1}{\delta}} = 0$. As a result $f(\delta)$ has a unique global maximizer at $\delta_0$ where $f(\delta_0) = x_1x_2$. Note in Fig. \ref{Gaussian_fig:fydelta}(a) that the zero of $y(\delta)$ coincides with the maximum of $f(\delta)$. 

\begin{figure}[t]
\centering
\scalebox{0.8}{\includegraphics{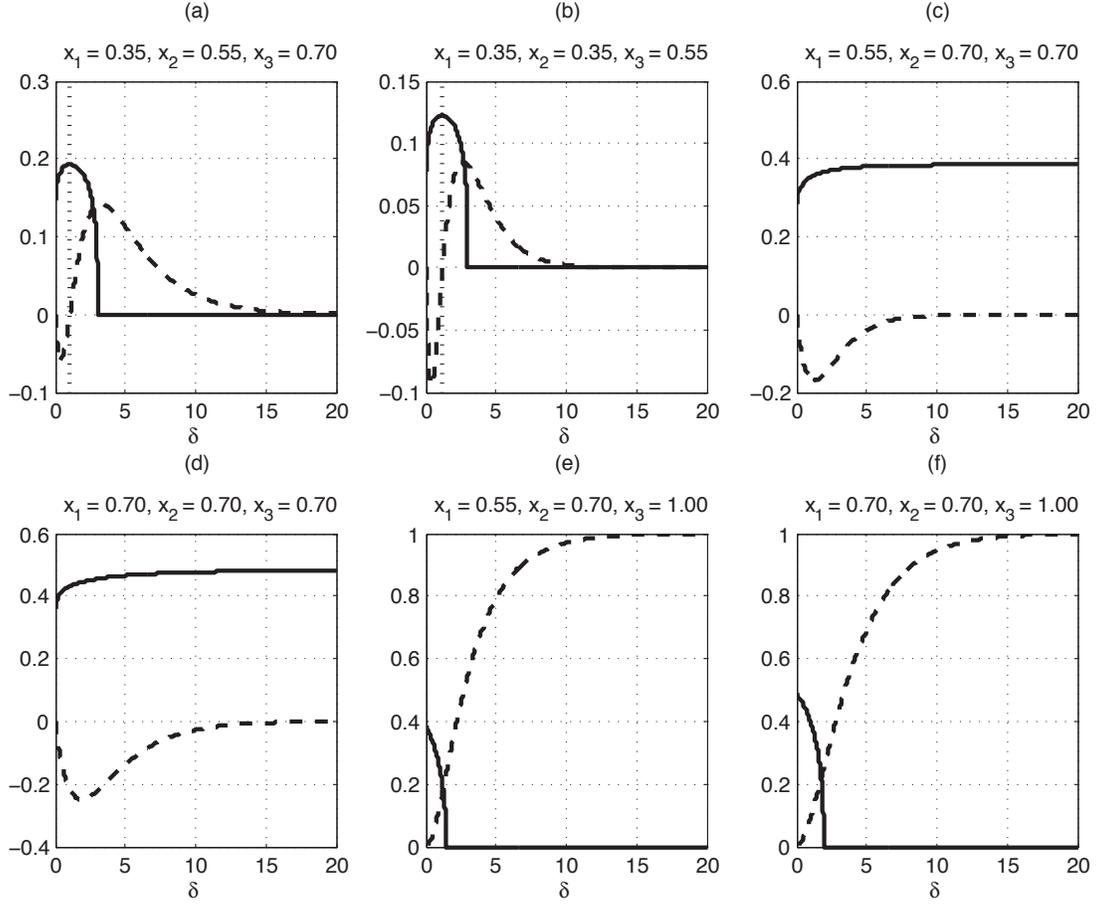}}
\caption{Functions $f(\delta)$ and $y(\delta)$ versus $\delta$. The solid line shows function $f(\delta)$ and the dashed line shows function $y(\delta)$.}
\label{Gaussian_fig:fydelta}
\end{figure}

\item $x_1=x_2<x_3<1$:

In this case we have $y(\delta) = x_1^{2\delta} + x_3^{\delta} - 2x_1^{\delta} $. Similar to the previous case we define $a(\delta) = \frac{y(\delta)}{x_3^{\delta}} = \left(\frac{x_1^2}{x_3}\right)^{\delta} + 1 - 2\left(\frac{x_1}{x_3}\right)^{\delta}$ and analyze the zeros of $a(\delta)$. Note that $\frac{da(\delta)}{d\delta} =  \left(\frac{{x_1}^2}{x_3}\right)^{\delta} \log \left(\frac{{x_1}^2}{x_3}\right) - 2 \left(\frac{x_1}{x_3}\right)^{\delta} \log \left(\frac{x_1}{x_3}\right)$ which has a unique zero at the point $\hat{\delta}$ given by,
\be
x_1^{\hat{\delta}} = \frac{ 2\log \frac{x_1}{x_3} }{  \log \frac{x_1^2}{x_3}}
\ee
Since $a(\delta) \approx \delta \log x_3 < 0 $ as $\delta \rightarrow 0^+$ and $a(\infty) \approx 1$ and again $a(\delta)$ is everywhere differnetiable, we obtain that $a(\delta)$ has exactly one zero (that is not at origin) at some $\delta_0$ where $\delta_0>\hat{\delta} >0$. Therefore $y(\delta)$ has also exactly one zero (that is not at origin) at $\delta_0$. Again we have $y(\delta) \approx \delta \log x_3 < 0 $ as $\delta \rightarrow 0^+$ and $y(\infty) \approx x_3^{\delta}>0$ and hence its behavior is again as the one depicted in Fig. \ref{Gaussian_fig:fydelta}(b). 

The analysis of $e(\delta)$ and $f(\delta)$ is similar to the last case. In particular we have, $f(0^+) \approx x_1^2x_3 $ and $e(\infty) \approx -\frac{1}{4}x_3^{2\delta}<0$ giving $f(\delta) = 0$ for $\delta \rightarrow \infty $. Hence $f(\delta)$ has a unique maximizer at $\delta_0$ such that $f(\delta_0) = x_1^2$. Again it is evident from Fig. \ref{Gaussian_fig:fydelta}(b) that the zero of $y(\delta)$ coincides with the maximum of $f(\delta)$.

\item $x_1<x_2=x_3<1$:

In this case $y(\delta) = (x_1x_3)^{\delta} - x_1^{\delta} \leq 0$. Calculating $\frac{dy(\delta)}{d\delta} = (x_1x_3)^{\delta} \log (x_1x_3) - x_1^{\delta} \log x_1 $ which has a unique zero at $\hat{\delta}$ given by,
\be
x_3^{\hat{\delta}} = \frac{ \log x_1}{ \log (x_1x_3)  }
\ee
Moreover as $\delta\rightarrow 0^+$ we have $y(\delta) \approx \delta \log x_3 \rightarrow 0^-$ and as $\delta \rightarrow \infty$, $y(\delta) \approx -x_1^\delta \rightarrow 0^-$. Therefore the behavior of $y(\delta)$ in this case is as shown in Fig. \ref{Gaussian_fig:fydelta}(c). 

The analysis of $e(\delta)$ and $f(\delta)$ is again similar to the previous 2 cases. However note that in this case the only zero of $y(\delta)$ is at origin and as a result we only need to consider the value of $f(\delta)$ at $0$ and infinity. Following a similar procedure as in the previous cases, we obtain $f(0^+) \approx x_1x_3^2$ and for $\delta \rightarrow \infty$ by replacing $x_2=x_3$ in (\ref{eqn:edeltataylor}) we get,
\be
e(\delta) \approx (x_1x_3)^\delta - \frac{1}{4}x_1^{2\delta} \approx (x_1x_3)^\delta
\ee
Therefore $f(\infty) \approx x_1x_3$, i.e., $f(\delta)$ approaches its global maximum at infinity.

\item $x_1=x_2=x_3<1$:

In this case, $y(\delta) = x_3^{2\delta} - x_3^{\delta}  \leq 0$ and we have $\frac{dy(\delta)}{d\delta} = 2x_3^{2\delta}\log x_3 - x_3^{\delta}\log x_3$ which again has a unique zero at $\hat{\delta}$ given by, $x_3^{\hat{\delta}} = \frac{1}{2} $. As in the previous case, $y(0^+) \approx \delta \log x_3 \rightarrow 0^-$ and $y(\infty) \approx -x_3^{\delta} \rightarrow 0^-$ and $y(\delta)$ behaves as in the previous case (Fig. \ref{Gaussian_fig:fydelta}(d)).

Since zeros of $y(\delta)$ happen at $0$ and infinity, we only need to calculate the value of $f(\delta)$ at 0 and infinity. For $\delta\rightarrow 0^+$ we have, $f(0^+)\approx x_3^3$. On the other hand by replacing $x_1=x_2=x_3$ in (\ref{eqn:edeltataylor}) we obtain as $\delta \rightarrow \infty$,
\be
e(\delta)\approx \frac{3}{4}x_3^{2\delta}
\ee
This yields $f(\delta) \approx \left(\frac{3}{4}\right)^\frac{1}{\delta} x_3^2 \approx x_3^{2}$ as $\delta \rightarrow \infty$. As a result $f(\delta)$ again approaches its global maximum at infinity.

\item $x_1<x_2 < x_3 = 1$:

In this case $y(\delta)$ simplifies to $y(\delta) = (x_1x_2)^{\delta} + 1 - x_1^{\delta} - x_2^{\delta} = (1-x_1^{\delta})(1-x_2^{\delta}) \geq 0$. Note that $y(0) = 0, y(\infty) \rightarrow 1$ and $y(\delta)$ is an increasing function. The behavior of $y(\delta)$ is shown in Fig. \ref{Gaussian_fig:fydelta}(e). 

Since the zero of $y(\delta)$ occurs at zero, we only need to evaluate $f(\delta)$ at zero for which we obtain $e(0^+)\approx 1+\delta \log (x_1x_2)$ and therefore $f(0^+) \approx x_1x_2$, i.e., the global maximum of $f(\delta)$ is at zero. 

\item $x_1 = x_2 <x_3 =1$:

In this case we have $y(\delta) = x_1^{2\delta} + 1 - 2x_1^{\delta} = (1-x_1^{\delta})^2 \geq 0$. Behavior of $y(\delta)$ is shown in Fig. \ref{Gaussian_fig:fydelta}(f).

For this scenario again we only need to care for $f(0)$ which can be easily obtained to be $f(0^+)\approx x_1^2$ which is the global maximum.

\item $x_1 \leq x_2 = x_3 = 1$:

Here we obtain $y(\delta) = 0$ a constant function. To evaluate $f(\delta)$ we do not need to use $y(\delta)$ in this case. In fact we have $f(\delta) = x_1$ which is a constant function as well and equal to its global maximum everywhere.

\end{enumerate}

Thus far we showed that $f(\delta)$ (except for the case when $x_1\leq x_2=x_3=1$ and $f(\delta)$ is a constant) it has a unique global maximizer at which $f(\delta) = \min_{i\neq j} x_ix_j$. Moreover in all these cases, if for some $\delta>0$ we have $y(\delta)\leq0$ it can be seen that maximum of $f(\delta)$ occurs for some $\delta_0>\delta$. Noting that $\tilde{y}$ as defined in the statement of the theorem is in fact $y(1)$, it immediately follows that if $\tilde{y}\leq0$ then $f(\delta)$ attains its global maximum for some $\delta\geq 1$.   
\hfill \qed



%



\bibliographystyle{IEEEbib}
\bibliography{ref}

\begin{thebibliography}{10}

\bibitem{hash2008}
B.~Hassibi and S.~Shadbakht,
\newblock ``The entropy region for three gaussian random variables,''
\newblock in {\em IEEE Int. Symp. on Inf. Theory (ISIT)}, 2008.

\bibitem{allertonShHa}
S.~Shadbakht and B.~Hassibi,
\newblock ``Cayley's hyperdeterminant, the principal minors of a symmetric
  matrix and the entropy region of 4 {G}aussian random variables,''
\newblock {\em 46th Annual Allerton Conf. on Comm., Control, and Computing,
  Monticello, IL}, 2008.

\bibitem{hashITW07}
B.~Hassibi and S.~Shadbakht,
\newblock ``Normalized entropy vectors, network information theory and convex
  optimization,''
\newblock in {\em Information theory workshop, Bergen, Norway}, 2007.

\bibitem{Yan-capacityregion}
X.~Yan, R.W. Yeung, and Z.~Zhang,
\newblock ``The capacity region for multi-source multi-sink network coding,''
\newblock in {\em IEEE Int. Symp. on Inf. Theory (ISIT)}, 2007, pp. 116--120.

\bibitem{framework}
R.W. Yeung,
\newblock ``A framework for linear information inequalities,''
\newblock {\em IEEE Trans. on Information Theory}, vol. 43, no. 6, pp.
  1924--1934, 1997.

\bibitem{ma07}
F.~Matus,
\newblock ``Infinitely many information inequalities,''
\newblock in {\em IEEE Int. Symp. on Inf. Theory (ISIT)}, 2007, pp. 41--44.

\bibitem{non-shannon}
Zhen Zhang and Raymond Yeung,
\newblock ``A non-shannon-type conditional inequality of information
  quantities,''
\newblock {\em IEEE Trans. on Information Theory}, vol. 43, no. 6, pp.
  1982--1986, 1997.

\bibitem{six-new}
R.~Dougherty, C.~Freiling, and K.~Zeger,
\newblock ``Six new non-shannon information inequalities,''
\newblock in {\em IEEE Int. Symp. on Inf. Theory (ISIT)}, 2006, pp. 233--236.

\bibitem{group}
T.H. Chan and R.W. Yeung,
\newblock ``On a relation between information inequalities and group theory,''
\newblock {\em IEEE Trans. on Information Theory}, vol. 48, no. 7, pp.
  1992--1995, 2002.

\bibitem{quasi-uniform}
T.H. Chan,
\newblock ``A combinatorial approach to information inequalities,''
\newblock {\em Communications In Information and Systems}, vol. 1, no. 3, pp.
  241--254, 2001.

\bibitem{makarychev}
K.~Makarychev, Y.~Makarychev, A.~Romashchenko, and N.~Vereshchagin,
\newblock ``A new class of non-shannon-type inequalities for entropies,''
\newblock {\em Communications In Information and Systems}, vol. 2, no. 2, pp.
  147--166, 2002.

\bibitem{zhang-new}
Z.~Zhang,
\newblock ``On a new non-{S}hannon type information inequality,''
\newblock {\em Communications In Inf. and Systems}, vol. 3, no. 1, pp. 47--60,
  2003.

\bibitem{hash2007}
B.~Hassibi and S.~Shadbakht,
\newblock ``On a construction of entropic vectors using lattice-generated
  distributions,''
\newblock in {\em IEEE Int. Symp. on Inf. Theory (ISIT)}, 2007, pp. 501--505.

\bibitem{characterization}
Z.~Zhang and R.~Yeung,
\newblock ``On characterization of entropy function via information
  inequalities,''
\newblock {\em IEEE Trans. on Information Theory}, vol. 44, no. 4, pp.
  1440--1452, 1998.

\bibitem{conditional-independence}
F.~Matus and M.~Studeny,
\newblock ``Conditional independences among four random variables {I},''
\newblock {\em Combin., Prob. Comput.}, vol. 4, pp. 269--278, 1995.

\bibitem{polymat}
S.~Fujishije,
\newblock ``Polymatroidal dependence structure of a set of random variables,''
\newblock {\em Information and Control}, vol. 39, pp. 55--72, 1978.

\bibitem{uniqueness}
T.~S. Han,
\newblock ``A uniqueness of shannon's information distance and related
  nonnegativity problems,''
\newblock {\em J. Comb.,Inform. Syst. Sci.}, vol. 6, no. 4, pp. 320--331, 1981.

\bibitem{matus-construction}
F.~Matus,
\newblock ``Two constructions on limits of entropy functions,''
\newblock {\em IEEE Trans. on Information Theory}, vol. 53, no. 1, pp.
  320--330, 2007.

\bibitem{chan-balanced}
T.~H. Chan,
\newblock ``Balanced information inequalities,''
\newblock {\em IEEE Trans. on Information Theory}, vol. 49, no. 12, pp.
  3261--3267, 2003.

\bibitem{hyperdet}
O.~Holtz and B.~Sturmfels,
\newblock ``Hyperdeterminantal relations among symmetric principal minors,''
\newblock {\em Journal of Algebra}, vol. 316, pp. 634--648, 2007.

\bibitem{cover-determinant}
T.~M. Cover and J.~A. Thomas,
\newblock ``Determinant inequalities via information theory,''
\newblock {\em SIAM J. Matrix Anal. Appl.}, vol. 9, no. 3, pp. 384--392, 1988.

\bibitem{Johnson00}
Shaun~M. Fallat and Charles~R. Johnson,
\newblock ``Determinantal inequalities: Ancient history and recent advances,''
\newblock {\em Contemporary Mathematics}, vol. 259, pp. 199--211, 2000.

\bibitem{Johnson93}
Charles~R. Johnson and Wayne~W. Barrett,
\newblock ``Determinantal inequalities for positive definite matrices,''
\newblock {\em Discrete Mathematics}, vol. 119, pp. 97--106, 1993.

\bibitem{johnson-bddminor-new}
H.~T. Hall and C.~R. Johnson,
\newblock ``Bounded ratios of products of principal minors of positive definite
  matrices,''
\newblock {\em arXiv:0806.2645v1}.

\bibitem{lnenicka-gauss-conditional}
R.~Ln\v{e}ni\v{c}ka and F.~Mat\'{u}\v{s},
\newblock ``On gaussian conditional independence structures,''
\newblock {\em Kybernetika}, vol. 43, no. 3, pp. 327--342, 2007.

\bibitem{GT-minor-computation}
K.~Griffin and M.~J. Tsatsomeros,
\newblock ``Principaal minors, part i: A method for computing all the principal
  minors of a matrix,''
\newblock {\em Linear Algebra and its applications}, vol. 419, no. 1, pp.
  107--124, 2006.

\bibitem{lnenicka-gauss}
Radim Lnenicka,
\newblock ``On the tightness of the {Z}hang-{Y}eung inequality for gaussian
  vectors,''
\newblock {\em Communications in information and systems}, vol. 3, no. 1, pp.
  41--46, 2003.

\bibitem{minor-assign}
K.~Griffin and Tsatsomeros M.~J,
\newblock ``Principal minors, part {II}: The principal minor assignment
  problem,''
\newblock {\em Linear Algebra and its applications}, vol. 419, pp. 125--171,
  2006.

\bibitem{ingleton}
A.~W. Ingleton,
\newblock ``Representation of matroids,''
\newblock in {\em Combinatorial mathematics and its applications, D.~Welsh, Ed.
  London: Academic Press}, 1971, pp. 149--167.

\bibitem{chan-abelian}
T.~H. Chan,
\newblock ``Capacity region for linear and {A}belian network codes,''
\newblock in {\em Information theory and applications workshop, San Diego, CA},
  2007.

\bibitem{chan-group}
T.~H. Chan,
\newblock ``Group characterizable entropy functions,''
\newblock in {\em IEEE Int. Symp. on Inf. Theory (ISIT)}, 2007, pp. 506--510.

\bibitem{cayley}
A.~Cayley,
\newblock ``On the theory of linear transformations,''
\newblock in {\em Cambridge Math. J. 4}, 1845, pp. 1--16.

\bibitem{gelfand}
I.~M. Gelfand, M.~M. Kapranov, and A.~V. Zelevinsky,
\newblock {\em Discriminants, resultants and multidimensional determinants},
\newblock Mathematics: Theory and Applications., 1994.

\bibitem{schalfli}
L.~Shl\"{a}fli,
\newblock ``\"{U}ber die resultante eines systemes mehrerer algebraischen
  gleichungen,''
\newblock {\em Denkschr. der Kaiserlichen Akad. der Wiss., Math-Naturwiss.
  Klasse}, vol. 4, 1852.

\bibitem{LT-hyperdet-invariants}
J-G Luque and J-Y Thibon,
\newblock ``Polynomial invariants of four qubits,''
\newblock {\em Phys. Rev. A 67, 042303}, 2003.

\bibitem{hyper4cube}
Peter Huggins, Bernd Sturmfels, Josephine Yu, and Debbie~S. Yuster,
\newblock ``The hyperdeterminant and triangulations of the 4-cube,''
\newblock {\em Mathematics of Computation}, vol. 77, no. 263, pp. 1653--1679,
  2008.

\bibitem{TW-hyper-integrable}
S.~P Tsarev and T.~Wolf,
\newblock ``Hyperdeterminants as integrable discrete systems,''
\newblock {\em J.~Phys. A: Math Theor}, vol. 42, no. 45, 2009.

\bibitem{PR-amoeba-newton}
M.~Passare and H.~Rullg{\aa}rd,
\newblock ``Amoebas. {M}onge-{A}mp\`{e}re measures, and triangulations of the
  newton polytope,''
\newblock {\em Duke Math. J.}, vol. 121, no. 3, pp. 481--507, 2004.

\end{thebibliography}

\end{document}